\documentclass[lettersize, journal]{IEEEtran}
\usepackage{amsmath,amsfonts}
\usepackage{enumitem}
\usepackage{bm}
\usepackage{dsfont}
\usepackage{amsthm}
\usepackage{booktabs}
\usepackage{amssymb}
\usepackage{algorithm}
\usepackage{algpseudocode}
\usepackage{array}
\usepackage[caption=false,font=normalsize,labelfont=rm,textfont=rm]{subfig}
\usepackage{textcomp}
\usepackage{stfloats}
\usepackage{url}
\usepackage{verbatim}
\usepackage{graphicx}
\usepackage{cite}
\usepackage{makecell}
\usepackage{xcolor}
\usepackage{footnote}
\makesavenoteenv{algorithm} 

\newtheorem{lemma}{Lemma}
\newtheorem{remark}{Remark}
\newtheorem{example}{Example}
\newtheorem{theorem}{Theorem}
\newtheorem{corollary}{Corollary}
\newtheorem{assumption}{Assumption}

\newtheorem{proposition}{Proposition}
\usepackage{hyperref}

\newcommand{\Exp}{\mathbb{E}}
\newcommand{\Var}{\mathbb{D}}
\newcommand{\Cov}{\bm{\mathrm{Cov}}}
\renewcommand{\d}{\,\mathrm{d}}
\DeclareMathOperator*{\argmin}{arg\,min} 
\DeclareMathOperator*{\argmax}{arg\,max} 
\DeclareMathOperator*{\rvec}{\mathrm{vec}}

\definecolor{darkgray}{gray}{0.4} 

\begin{document}


\title{Nonlinear Bayesian Filtering with \\ \textbf{N}atural Gr\textbf{a}dient Gaussia\textbf{n} Appr\textbf{o}ximation}

\author{Wenhan Cao, Tianyi Zhang, Zeju Sun, Chang Liu, \\
Stephen S.-T. Yau, \IEEEmembership{Life Fellow, IEEE}, and Shengbo Eben Li, \IEEEmembership{Senior Member, IEEE} 
\thanks{Wenhan Cao is with the School of Vehicle and Mobility, Tsinghua University, Beijing, China (e-mail: cwh19@mails.tsinghua.edu.cn).}

\thanks{Tianyi Zhang is with the State Key Laboratory of Intelligent
Green Vehicle and Mobility, Tsinghua University, Beijing, China (e-mail: zhangtia24@mails.tsinghua.edu.cn).}

\thanks{Zeju Sun is with the Department of Mathematical Sciences, Tsinghua University, Beijing, China (e-mail: szj20@mails.tsinghua.edu.cn).}

\thanks{Chang Liu is with the Department of Advanced Manufacturing
and Robotics, College of Engineering, Peking University, Beijing, China (e-mail: changliucoe@pku.edu.cn).
}

\thanks{Stephen S.-T. Yau is with Beijing Institute of Mathematical Sciences and Applications (BIMSA) and the Department of
Mathematical Sciences, Tsinghua University, Beijing, China (e-mail: yau@uic.edu).}

\thanks{Shengbo Eben Li is with the School of Vehicle and Mobility and College
of Artificial Intelligence, Tsinghua University, Beijing, China (e-mail: lishbo@tsinghua.edu.cn).}

\thanks{Corresponding Author: Shengbo Eben Li}
}

\markboth{Journal of \LaTeX\ Class Files,~Vol.~14, No.~8, August~2021}%
{Shell \MakeLowercase{\textit{et al.}}: A Sample Article Using IEEEtran.cls for IEEE Journals}

\maketitle

\begin{abstract} 
Practical Bayes filters often assume the state distribution of each time step to be Gaussian for computational tractability, resulting in the so-called Gaussian filters. When facing nonlinear systems, Gaussian filters  such as extended Kalman filter (EKF) or unscented Kalman filter (UKF) typically rely on  certain linearization techniques, which can introduce large estimation errors.
To address this issue, this paper reconstructs the prediction and update steps of Gaussian filtering as solutions to two distinct optimization problems, whose optimal conditions are found to have analytical forms from Stein's lemma. It is observed that the stationary point for the prediction step requires calculating the first two moments of the prior distribution, which is equivalent to that step in existing moment-matching filters.  
In the update step, instead of linearizing the model to approximate the stationary points, we propose an iterative approach to directly minimize the update step's objective to avoid linearization errors. For the purpose of performing the steepest descent on the Gaussian manifold, we derive its natural gradient that leverages Fisher information matrix to adjust the gradient direction, accounting for the curvature of the parameter space.  Combining this  update step with moment matching in the prediction step, we introduce a new iterative filter for nonlinear systems called \textit{N}atural Gr\textit{a}dient Gaussia\textit{n} Appr\textit{o}ximation filter, or NANO filter for short. 
We prove that NANO filter locally converges to the optimal Gaussian approximation at each time step. Furthermore, the estimation error is proven exponentially
bounded for nearly linear measurement equation and low
noise levels through constructing a supermartingale-like
property across consecutive time steps. 
Real-world experiments demonstrate that, compared to popular Gaussian filters such as EKF, UKF, iterated EKF, and posterior linearization filter,  NANO filter reduces the average root mean square error by approximately 45\% while maintaining a comparable computational burden.


\end{abstract}


\section{Introduction} \label{sec.1}

\IEEEPARstart{S}{tate} estimation of dynamical systems is a timely topic in fields such as astrophysics, robotics, power systems, manufacturing, and transportation. The most comprehensive framework for state estimation is Bayesian filtering, which targets the distribution of the current state given the available measurements to date. This state distribution, referred to as the posterior distribution, is often calculated iteratively through two steps, i.e., prediction and update. The prediction step uses the Chapman-Kolmogorov equation to forward-predict the state distribution using the transition probability, a conditional probability that describes the time evolution of the system state. Based on the prior distribution acquired after prediction, the update step applies Bayes' theorem to update the prior using the measurement probability, a conditional probability that describes the relationship between noisy measurements and the true state \cite{chen2003bayesian}.

For linear Gaussian systems, directly applying Bayesian filtering results in the well-known Kalman filter (KF) \cite{kalman1960new}, which analytically computes the Gaussian posterior through recursive updates of its mean and covariance. This analytical form relies on the closure property of Gaussian under linear transformations and the conjugate property of Gaussian under conditioning. Unfortunately, such an elegant structure of KF does not exist for nonlinear systems, as the closure and conjugate properties only hold in the linear Gaussian case.

Therefore, if the system is nonlinear or non-Gaussian, finding the exact solution of Bayesian filtering is often unattainable. In this case,  designing an appropriate approximation of state distribution becomes a key step in Bayesian filtering. To this effect, the particle filter (PF) approximates the probability density function with a set of discrete particles, each representing a possible state of the system. These particles, when weighted and summed, form a discrete approximation of the continuous probability density function. 
Although PF can 
provide asymptotically optimal approximations, it requires a large number of particles, which results in significant computational demands and limits its application in practical systems \cite{del1997nonlinear,liu1998sequential}. In contrast to approximating the posterior distribution with samples, an alternative choice is to approximate it as Gaussian distribution in each time step, which leads to the Gaussian filter family. 
Compared to PF, Gaussian filters offer higher computational efficiency and have become by far the most popular family to date \cite{thrun2002probabilistic}.

It turns out that the design philosophy of Gaussian filters generally consists of two steps: (i) approximate the nonlinear models to linear forms with additive Gaussian noise, and (ii) perform KF based on this linear Gaussian model. Within this framework, the differences between Gaussian filters are primarily attributed to different linearization techniques.
The earliest technique involves directly using the Taylor series approximation to linearize the nonlinear function.  Typically, the first-order expansion is employed to avoid the occurrence of tensors in high-order expansions, exemplified by the well-known extended Kalman filter (EKF) \cite{smith1962application,mcelhoe1966assessment}. Originally developed by NASA for navigation tasks, EKF linearizes the nonlinear function around the state estimate and performs KF using the resulting affine-form system equation. Building on this foundation, iterated extended Kalman filter (IEKF)  \cite{gelb1974applied}  repeatedly performs linearization at each updated approximation of the posterior mean instead of the prior mean as in EKF. 
In fact, the iterative process of this method is a Gauss-Newton iteration that essentially linearizes the model around the maximum a posteriori estimate of the state \cite{bell1993iterated}.

A clear drawback of directly linearizing the state space model is its inability to capture the second moment, namely the covariance, after a nonlinear transformation, since Taylor-series linearization only utilizes the mean. To address this issue, a fundamentally different approach is to directly match the transformed mean and covariance, known as moment-matching. Essentially, moment-matching can be viewed as statistical linear regression (SLR), a linearization technique that obtains optimal affine representation of the nonlinear system with parameters minimizing the expected 
regression loss \cite{sarkka2023bayesian}.
This so-called SLR method requires computing several integrals over Gaussian distributions. In particular, the use of established numerical integration methods such as unscented transform, Gauss–Hermite integration, and spherical cubature integration  underpin the design of  unscented Kalman filter (UKF) \cite{julier1995new}, Gauss–Hermite KF \cite{arasaratnam2007discrete}, and cubature KF \cite{arasaratnam2009cubature}. 
In addition, a new algorithm called posterior linearization filter 
(PLF) is proposed to perform SLR at the posterior rather than the prior during Bayesian updates to improve filtering accuracy \cite{sarkka2023bayesian}. 

This framework of first linearizing the model and then performing KF is termed \textit{enabling approximation}  \cite{sarkka2023bayesian}. Although it has been widely applied since the 1960s and may seem like a natural choice, one critical question cannot be overlooked: Is this \textit{enabling approximation} framework sufficient to find the optimal Gaussian approximation of Bayesian filtering? Unfortunately, no existing works discuss this issue. In this paper, we argue that applying this framework in the update step may not yield an exact solution for nonlinear Gaussian filtering. This contrasts with the prediction step, where performing \textit{enabling approximation} using moment-matching filters \cite{julier1995new,arasaratnam2007discrete,arasaratnam2009cubature} is a proper choice. 
To this effect, we propose a new method called \textbf{N}atural gr\textbf{A}dient Gaussia\textbf{N} appr\textbf{O}ximation (NANO) filter that applies natural gradient in the update step to find the exact solution of Gaussian approximation.
Specifically, the contributions of this paper are summarized as follows:
\begin{itemize}

\item We interpret the prediction and update steps of Bayesian filtering as solutions to two distinct optimization problems. This new  perspective allows us to define optimal Gaussian approximation and identify its corresponding extremum conditions. Leveraging the Stein’s lemma, we derive that the stationary point for prediction step has an explicit form, involving the calculation of the first two moments of prior distribution. This analytical form is implementable via moment-matching filters \cite{julier1995new, arasaratnam2007discrete, arasaratnam2009cubature}. In contrast, the stationary point for update step is  characterized by two interdependent equations, which generally has no analytical root. A special case occurs in the linear Gaussian systems, where these two equations decouple and become linear, resulting in analytical solutions for KF.
For nonlinear systems, those \textit{enabling approximation}-based  Gaussian filters need to perform certain linearization technique, which inevitably introduce large
estimation errors.

\item
To address linearization errors in the update step, we derive natural gradient iteration to minimize the optimization cost. By leveraging the Fisher information matrix that captures the curvature of the parameter space, the gradient direction is adjusted to perform the steepest descent on the Gaussian manifold. By combining this optimization procedure in the update step with moment matching in the prediction step, we develop a new iterative filter for nonlinear systems, namely NANO filter. We demonstrate that KF is equivalent to  a single iteration of NANO filter for linear Gaussian systems, providing a fresh understanding of Kalman filtering.

\item 
We prove that the NANO filter locally converges to the optimal Gaussian approximation at each time step, with accuracy up to a second-order remainder in the Taylor expansion.
Additionally, the estimation error is proven exponentially bounded for nearly linear measurement equation and low noise levels through contructing a supermartingale-like property across consecutive time steps.
We further show that NANO filter can naturally  extend the Bayesian posterior to Gibbs posterior, enabling the use of more flexible loss functions to enhance robustness against outliers in measurement data. On its basis, three robust variants of  NANO filter are introduced, each employing different robust loss function designs. Simulations and real-world experiments demonstrate that NANO filter and its robust variants significantly outperforms popular Gaussian filters, such as EKF, UKF, IEKF, and PLF for nonlinear systems.
\end{itemize}
The remainder of this paper is structured as follows: Section \ref{sec.2} formulates the problem, and Section \ref{sec.Gaussian approximation} explores the optimal Gaussian approximation. Section \ref{sec.NANO} introduces the NANO filter algorithm, followed by theoretical analysis in Section \ref{sec.theoretical analysis}. Section \ref{sec.robust variants} proposes the robust variants of NANO filter. Section \ref{sec.discussion} provides a discussion on the proposed algorithm. Simulations and experiments are presented in Section \ref{sec.experiment}.

\noindent \textbf{Notation:} All vectors are considered as column vectors. The symbol \( D_{\mathrm{KL}}(p \| q) = \mathbb{E}_{p}\left\{\log{\frac{p}{q}}\right\} \) denotes the Kullback-Leibler (KL) divergence between two probability distributions \( p \) and \( q \). Unless otherwise specified, \( \|x\| \) refers to the \( \ell_2 \)-norm of the vector \( x \), defined as \( \|x\| = \sqrt{x^\top x} \). The notation \( \mathcal{N}(x; \mu, \Sigma) \) represents the Gaussian probability density function for the variable \( x \) with mean \( \mu \) and covariance matrix \( \Sigma \). For simplicity, this may be abbreviated as \( \mathcal{N}(\mu, \Sigma) \). Furthermore, \( \mathbb{E} \{ x \} \), \( \mathbb{D} \{ x \} \), and \( \mathrm{Cov}(x, y) \) denote the expectation, variance, and covariance, respectively. The notation \( A \leq B \) (\( A < B \)) for symmetric matrices \( A \) and \( B \) indicates that \( B - A \) is positive (semi-)definite. The symbol \( \otimes \) denotes the Kronecker product. The notation \(\rvec(A)\) refers to the vector obtained by stacking the columns of matrix \( A \) into a single column vector. The trace of matrix \( A \) is denoted as \( \mathrm{Tr}(A) \). The notation \( \mathbb{I}_{n \times n} \) represents the \( n \)-dimensional identity matrix.

\section{Problem Statement}\label{sec.2}
Consider the following nonlinear discrete-time stochastic system:
\begin{equation}\label{eq.SSM}
\begin{aligned}
x_{t+1} = f(x_t)+\xi_t,
\\
y_t = g(x_t) + \zeta_t, 
\end{aligned}
\end{equation}
where $x_t \in \mathbb{R}^n$ is the system state, $y_t \in \mathbb{R}^m$ is the noisy measurement. The function ${f}: \mathbb{R}^n \to \mathbb{R}^n$ is referred to as the transition function, while ${g}: \mathbb{R}^n \to \mathbb{R}^m$ is called measurement function; ${\xi}_t$ denotes process noise, and ${\zeta}_t$ represents measurement noise. Typically, the initial state \(x_0\), the process noise \(\{\xi_t\}\), and the measurement noise \(\{\zeta_t\}\) are all mutually independent, with the process and measurement noises being independent and identically distributed across time. 
This state-space model description in \eqref{eq.SSM} can be represented as a hidden Markov model: 
\begin{equation}\label{eq.HMM}
\begin{aligned}
x_{0} &\sim p(x_0),
\\
x_{t} &\sim p(x_t|x_{t-1}),
\\
y_{t} &\sim p(y_t|x_{t}).
\end{aligned}  
\end{equation}
Here, $p(x_t|x_{t-1})$ and $p(y_{t}|x_{t})$ are the transition and output probabilities respectively while $p(x_0)$ denotes the initial state distribution.
In essence, \eqref{eq.HMM} and \eqref{eq.SSM} are different representations of the same system model. For example, consider the transition model \(x_t = Ax_{t-1} + \xi_{t-1}\), where \(\xi_{t-1} \sim \mathcal{N}(\xi_{t-1}; 0, Q)\), with \(Q\) denoting the covariance matrix of the process noise. This can be equivalently described by the transition probability \(p(x_t|x_{t-1}) = \mathcal{N}(x_t; Ax_{t-1}, Q)\).

The objective of state estimation is to recover the system state $x_t$ from noisy measurements $y_{t}$. Typically, finding the optimal estimate involves two key steps: calculating the posterior distribution $p(x_t|y_{1:t})$, and determining the optimal estimate $\hat{x}_{t|t}$ using the posterior distribution. A principled framework for posterior distribution calculation is Bayesian filtering, which computes $p(x_t|y_{1:t})$ recursively through two steps:
\begin{subequations}\label{eq.BF}
\begin{align}
p(x_t|y_{1:t-1}) &= \int p(x_t|x_{t-1}) p(x_{t-1}|y_{1:t-1}) \d x_{t-1}, \label{eq.prediction}
\\
p(x_t|y_{1:t}) &= \frac{p(y_t|x_t) p(x_t|y_{1:t-1})}{\int p(y_t|x_t) p(x_t|y_{1:t-1}) \d x_t }. 
\label{eq.update}
\end{align}   
\end{subequations}
Here, \eqref{eq.prediction} is called 
prediction step while \eqref{eq.update} is called update step. The former, known as the Chapman-Kolmogorov equation, utilizes the transition probability $p(x_t|x_{t-1})$ to predict the prior distribution $p(x_t|y_{1:t-1})$. 
Based on this prior, the update step leverages the Bayes formula to calculate the posterior distribution $p(x_t|y_{1:t})$, where the output probability serves as 
likelihood function. After calculating $p(x_t|y_{1:t})$,  common estimation criteria like the minimum mean square error or maximum a posteriori estimation can be used to determine the optimal estimate.

As discussed in Section \ref{sec.1}, neither \eqref{eq.prediction} nor \eqref{eq.update} can be calculated analytically when the transition or output probabilities are non-linear or non-Gaussian, requiring the approximation of state distributions in practice. Given the high computational burden of discrete approximations like PF, Gaussian approximations are widely adopted in industrial applications. More specifically, in Gaussian filters, both the prior and posterior distributions are approximated by Gaussian distributions:
\begin{equation}\label{eq.Gaussian approximation}
\begin{aligned}
p(x_t|y_{1:t-1}) &\approx \mathcal{N}(x_t; \hat{x}_{t|t-1}, P_{t|t-1}), \\ p(x_t|y_{1:t}) &\approx \mathcal{N}(x_t; \hat{x}_{t|t}, P_{t|t}),   
\end{aligned}    
\end{equation}
which is usually achieved through
\textit{enabling approximation} framework. In particular, this framework first approximates the transition or output probabilities with linear Gaussian models
and then performs the well-known KF. Due to the closure property of Gaussians under linear transformations and the conjugate property of Gaussians under conditioning, the resulting prior and posterior distributions are naturally approximated as Gaussian.
The linearization techniques in existing Gaussian filters can be categorized into two types: Taylor series expansion and stochastic linear regression. As summarized in Table \ref{tab.linearization techniques}, the former uses the Jacobian matrix $g'(\bar{x})$ to provide local affine approximations of system models at point $\bar{x}$, while the latter minimizes the expectation of the square of the  residual $y - Ax - b$ to find optimal linear parameters.
\begin{table*}
    \centering
    \caption{Linearization techniques for existing Gaussian filters}
    \begin{tabular}{ccc}
    \toprule
     Linearization Technique & \makecell{Basic Principle:   $\mathcal{N}(y; g(x), \Sigma) \approx \mathcal{N}(y; Ax + b, \Lambda)$} & Representative Algorithms \\
         \midrule
     Taylor Series Expansion & 
     \makecell{$g(x) = g(\bar{x}) + g'(\bar{x})(x - \bar{x})$ \\ $\Downarrow$ \\ $A = g'(\bar{x}), b = g(\bar{x}) - g'(\bar{x})\bar{x}, \Lambda = \Sigma$}
     & EKF\cite{smith1962application}, IEKF\cite{gelb1974applied}
     \\
         \midrule
     Statistical Linear Regression  & \makecell{ $\argmin_{A,b} \Exp_{x}\left\{ (y-Ax-b)^{\top}(y-Ax-b) \right\}$ \\ $\Downarrow$ \\ $A = \Cov(x,y)^{\top} \Var(x), b = \Exp\left\{ y\right\} - A \Exp\left\{ x\right\}, \Lambda = \Var\left\{ y\right\} - A \Var\left\{ x\right\} A^{\top} $}
     & \makecell{UKF \cite{julier1995new}, Gauss–Hermite KF \cite{arasaratnam2007discrete},\\ cubature KF \cite{arasaratnam2009cubature}, and PLF\cite{garcia2015posterior}}
     \\
    \bottomrule
    \end{tabular}

    \label{tab.linearization techniques}
\end{table*}

While these techniques offer practical methods for Gaussian approximation, it still remains unclear whether this \textit{enabling approximation} framework truly achieve an optimal Gaussian approximation for Bayesian filtering. Unfortunately, there is even no clear way to formally judge whether a Gaussian approximation is optimal. This paper aims to establish a framework for defining and identifying the optimal Gaussian approximation for Bayesian filtering. Specifically, we seek to address the following three questions:

\begin{enumerate}[label=\textbf{Q\arabic*:}]
\item What defines the optimal Gaussian approximation in Bayesian filtering, and what conditions must it meet?
\item How can we design a filter that effectively achieves this optimal approximation?
\item What theoretical guarantees can we provide for the algorithm's convergence and the boundedness of the estimation error?
\end{enumerate}

\section{Optimal Gaussian Approximation for Bayesian Filtering}\label{sec.Gaussian approximation}
In this section, we will address \textbf{Q1} by exploring the necessary conditions for determining the optimal Gaussian approximation for Bayesian filtering.
\subsection{Optimization Viewpoint of Bayesian Filtering}
Inspired by the optimization-centric view on Bayes's rule \cite{knoblauch2019generalized}, we show that the prior and posterior distributions can be interpreted as solutions of two variational problems, as shown in the subsequent proposition. 

\begin{proposition}[Variational Problems for Bayesian filtering]\label{prop.optimization}
The prior distribution can be regraded as the maxmizer of an variational problem:
\begin{equation}\label{eq.BF prediction optimizaton}
\begin{aligned}
p(x_t|y_{1:t-1}) = \argmax_{q(x_t)} 
\Exp_{\substack{p(x_{t-1}|y_{1:t-1}) \\ p(x_t|x_{t-1})}}\left\{
\log{q(x_t)}
\right\}.
\end{aligned}
\end{equation} Similarly, the posterior distribution can be regarded as the minimizer of a functional:
\begin{equation}\label{eq.BF update optimizaton}
\begin{aligned}
p(x_t|y_{1:t}) =& \argmin_{q(x_t)} \big\{\Exp_{q(x_t)} \left\{-\log{p(y_t|{x}_t)}\right\}\\
&+ D_{\mathrm{KL}}\left(q(x_t) || p(x_t|y_{1:t-1})\right) \big\}.   
\end{aligned}
\end{equation}  
Note that in both \eqref{eq.BF prediction optimizaton} and \eqref{eq.BF update optimizaton}, $q: \mathbb{R}^n \to \mathbb{R}$ represents the candidate density function. Besides, we use the notation $ \Exp_{\substack{p(x)\\p(y)}}\left\{f(x,y) \right\}$ to represent the expectation of $f(x,y)$ with respect to both distributions $p(x)$ and $p(y)$, i.e., $\Exp_{\substack{p(x)\\p(y)}}\left\{f(x,y) \right\} \triangleq \Exp_{p(x)} \Exp_{p(y)} \left\{f(x,y) \right\}$. 
\end{proposition}
\begin{proof}
By recognizing that the focus is solely on the extremizer and not on the objective value itself, it follows that for any constant 
$Z>0$, the right-hand side of \eqref{eq.BF prediction optimizaton} can be rewritten as
\begin{equation}
\begin{aligned}\label{eq.BF prediction optimiation 1}
\argmin_{q(x_t)}\left\{\Exp_{\int p(x_t|x_{t-1}) p(x_{t-1}|y_{1:t-1}) \d x_{t-1}}\left\{\log{\frac{
Z}{q(x_t)}} \right\} \right\}.
\end{aligned}
\end{equation}
If we choose $Z = \int p(x_t|x_{t-1}) p(x_{t-1}|y_{1:t-1}) \d x_{t-1}$  and recall that the KL divergence reaches its minimum value uniquely when its arguments are identical, we find that the solution to \eqref{eq.BF prediction optimiation 1} is indeed  $$\int p(x_t|x_{t-1}) p(x_{t-1}|y_{1:t-1}) \d x_{t-1},$$ which corresponds exactly to the prior distribution defined in \eqref{eq.prediction}. 
Similarly, the right-hand side of \eqref{eq.BF update optimizaton} is equal to
\begin{equation}
\begin{aligned}\label{eq.BF update optimiation 1}
\argmin_{q(x_t)}\left\{\Exp_{q(x_t)}\left\{\log{\frac{q(x_t)}{p(y_t|x_t)p(x_t|y_{1:t-1})/Z}} \right\} \right\}. 
\end{aligned}
\end{equation}   
If we choose $Z = \int p(y_t|x_t)p(x_t|y_{1:t-1}) \d x_t$, the solution to \eqref{eq.BF update optimiation 1} is $$\frac{p(y_t|x_t)p(x_t|y_{1:t-1})}{\int p(y_t | x_t){p}(x_t|y_{1:t-1})\, \mathrm{d}x_t},$$ which is precisely the posterior distribution defined in \eqref{eq.update}.  
\end{proof}


This proposition shows that the prior distribution \( p(x_t|y_{1:t-1}) \) can be viewed as the solution to a variational problem that maximizes the expected logarithm of a candidate density \( q(x_t) \) over the joint distribution of the previous state and the transition probability. This reflects the idea that the prior is derived by considering all possible transitions from the previous state, and selecting the distribution that maximizes the expected log-density under these transitions. Similarly, the posterior distribution \( p(x_t|y_{1:t}) \) is the solution to a variational problem that minimizes a cost combining the expected negative log-likelihood of the measurement model, \( -\log{p(y_t|{x}_t)} \), with the KL divergence between the candidate density and the prior distribution \( D_{\mathrm{KL}}(q(x_t) || p(x_t|y_{1:t-1})) \). This captures the Bayesian update process, where the posterior distribution adjusts the prior distribution based on new measurements to balance prior knowledge with new information.

In both \eqref{eq.BF prediction optimizaton} and \eqref{eq.BF update optimizaton}, there are no constraints on the candidate distribution, meaning we seek an optimal candidate distribution over the entire probability space. However, such variational problems generally lack analytical solutions. A tractable approach to solve this problem is to restrict the candidate distribution to a parameterizable family of distributions, with the Gaussian distribution family being the most commonly used \cite{kingma2013auto}. Therefore, by utilizing Gaussian approximations in \eqref{eq.Gaussian approximation}, the two variational problems \eqref{eq.BF prediction optimizaton} and \eqref{eq.BF update optimizaton} that depict Bayesian filtering can be reduced to two optimization problems:
\begin{subequations}\label{eq.Gaussian optimization}
\begin{align}
\hat{x}_{t|t-1}, P_{t|t-1} =&\argmax_{\hat{x}_{t}, P_t}
L(\hat{x}_t, P_t), \nonumber
\\
L(\hat{x}_t, P_t)=&\Exp_{\substack{\mathcal{N}(x_{t-1}; \hat{x}_{t-1|t-1}, P_{t-1|t-1}) \\ p(x_t|x_{t-1})}}  \left\{
\log{\mathcal{N}(x_t; \hat{x}_t, P_t)}\right\}, \label{eq.Gaussian optimization prediction}
\\
\hat{x}_{t|t}, P_{t|t} =& \argmin_{\hat{x}_t, P_t} J(\hat{x}_t, P_t), \nonumber
\\
J(\hat{x}_t, P_t) =& D_{\mathrm{KL}}\left(\mathcal{N}(x_t; \hat{x}_t, P_t) || \mathcal{N}(x_{t}; \hat{x}_{t|t-1}, P_{t|t-1})\right)   \nonumber
\\
&-\Exp_{\mathcal{N}(x_t; \hat{x}_t, P_t)} \big\{\log{p(y_t|{x}_t)}\big\}.\label{eq.Gaussian optimization update}
\end{align}    
\end{subequations}
Here, $L(\hat{x}_t, P_t)$ is called prediction cost while $J(\hat{x}_t, P_t)$ is called update cost. Compared with \eqref{eq.BF prediction optimizaton} and \eqref{eq.BF update optimizaton}, \eqref{eq.Gaussian optimization prediction} and \eqref{eq.Gaussian optimization update} transform the generally unsolvable variational problems to optimization problems for Gaussian parameters. This allows us to find the optimal Gaussian approximation for Bayesian filtering by studying the optimality conditions of \eqref{eq.Gaussian optimization}.

\subsection{Structure of Stationary Points of \texorpdfstring{$L(\hat{x}_t, P_t)$}{L(x\_t, P\_t)}}\label{sec.stationary points of L}
To establish the optimality conditions for problems \eqref{eq.Gaussian optimization prediction} and \eqref{eq.Gaussian optimization update}, we examine the stationary points of $L(\hat{x}_t, P_t)$ and $J(\hat{x}_t, P_t)$. Interestingly, the extreme condition of the former reduces to a straightforward moment-matching equation, while that of the latter results in two mutually coupled implicit equations, whose roots are generally intractable. The following lemma helps elucidate the structure of $L(\hat{x}_t, P_t)$:
\begin{lemma}[Stationary Points for Maximum Gaussian Likelihood]\label{lemma.Gaussian MLE}
For probability density function $p(x)$, the stationary points of a maximum expected Gaussian likelihood problem 
\begin{equation}\label{eq.maximum expected Gaussian likelihood}
\mu^*, \Sigma^* = \argmax_{\mu, \Sigma} \Exp_{p(x)} \left\{ \log \mathcal{N}(x; \mu, \Sigma) \right\},    
\end{equation}
can be written as
\begin{subequations}
\begin{align}
\mu^* &= \Exp_{p(x)}\left\{ x \right\}, \label{eq.moment matching mean}
\\
\Sigma^* &= \Exp_{p(x)} \left\{ (x-\mu^*) (x-\mu^*)^{\top} \right\}\nonumber
\\
&= \Exp_{p(x)} \left\{ x x^{\top} \right\} - \mu^* {\mu^*}^{\top}. \label{eq.moment matching variance}    
\end{align}
\end{subequations}
\end{lemma}
\begin{proof}
The log-likelihood function of a Gaussian distribution $\mathcal{N}(x; \mu, \Sigma)$ is given by
\begin{equation}\nonumber
\begin{aligned}
\log \mathcal{N}(x; \mu, \Sigma) =& -\frac{1}{2}\big[ \log |\Sigma| 
\\
& + (x - \mu)^\top \Sigma^{-1} (x - \mu) + \log (2\pi)^n \big],    
\end{aligned}
\end{equation}
where $n$ is the dimensionality of $x$. Taking the expectation with respect to the probability density function $p(x)$, we obtain
\begin{equation}\nonumber
\begin{aligned}
\Exp_{p(x)}\left\{ \log \mathcal{N}(x; \mu, \Sigma) \right\} =& -\frac{1}{2} \big[ \log |(2\pi)^n \Sigma| \\
&+ \Exp_{p(x)} \left\{ (x - \mu)^\top \Sigma^{-1} (x - \mu) \right\} \big].
\end{aligned}
\end{equation}

We first compute the gradient of this expectation with respect to $\mu$. The relevant term involving $\mu$ is the quadratic form $(x - \mu)^\top \Sigma^{-1} (x - \mu)$. Differentiating this term with respect to $\mu$ gives
\begin{equation}\nonumber
\frac{\partial}{\partial \mu} \Exp_{p(x)} \left\{ (x - \mu)^\top \Sigma^{-1} (x - \mu) \right\} = -2 \Sigma^{-1} \Exp_{p(x)} \left\{ x - \mu \right\}.
\end{equation}
Setting this derivative equal to zero, we obtain \eqref{eq.moment matching mean}. Next, we compute the gradient with respect to $\Sigma$. The relevant terms involving $\Sigma$ are the log-determinant $\log |\Sigma|$ and the quadratic form $(x - \mu)^\top \Sigma^{-1} (x - \mu)$. The gradient of the log-determinant with respect to $\Sigma$ is
$\frac{\partial}{\partial \Sigma} \log |\Sigma| = \Sigma^{-1}$. The gradient of the quadratic form is
\begin{equation}\nonumber
\begin{aligned}
&\frac{\partial}{\partial \Sigma} \Exp_{p(x)} \left\{ (x - \mu)^\top \Sigma^{-1} (x - \mu) \right\} \\
=& -\Sigma^{-1} \Exp_{p(x)} \left\{ (x - \mu)(x - \mu)^\top \right\} \Sigma^{-1}.    
\end{aligned}
\end{equation}
Combining these results, we have
\begin{equation}\nonumber
\begin{aligned}
&\frac{\partial}{\partial \Sigma} \Exp_{p(x)}\left\{ \log \mathcal{N}(x; \mu, \Sigma) \right\} 
\\
=& -\frac{1}{2} \left[ \Sigma^{-1} - \Sigma^{-1} \Exp_{p(x)} \left\{ (x - \mu)(x - \mu)^\top \right\} \Sigma^{-1} \right].    
\end{aligned}
\end{equation}
Setting this derivative equal to zero gives
\begin{equation}\nonumber
\Sigma^{-1} = \Sigma^{-1} \Exp_{p(x)} \left\{ (x - \mu)(x - \mu)^\top \right\} \Sigma^{-1},
\end{equation}
which implies \eqref{eq.moment matching variance}.
\end{proof}
This lemma states that the stationary points for maximizing the expected Gaussian likelihood are achieved when the Gaussian distribution matches the mean and variance of the given distribution. Therefore, as a special case of \eqref{eq.maximum expected Gaussian likelihood}, solving \eqref{eq.Gaussian optimization prediction} requires to match the first and second order moment:
\begin{equation}\label{eq.BF prediction solution}
\begin{aligned}
\hat{x}_{t|t-1} &= \Exp_{\substack{\mathcal{N}(x_{t-1}; \hat{x}_{t-1|t-1}, P_{t-1|t-1}) \\ p(x_t|x_{t-1})}}  \left\{x_t \right\},
\\
P_{t|t-1} &= \Exp_{\substack{\mathcal{N}(x_{t-1}; \hat{x}_{t-1|t-1}, P_{t-1|t-1}) \\ p(x_t|x_{t-1})}}  \left\{x_t x_t^{\top} \right\} - \hat{x}_{t|t-1} \hat{x}_{t|t-1}^{\top}.
\end{aligned}    
\end{equation}
For the system obeys the transition model in \eqref{eq.SSM} with $\xi_t$ being the zero mean process noise, \eqref{eq.BF prediction solution} can be further expressed as
\begin{equation}\label{eq.BF prediction solution 1}
\begin{aligned}
\hat{x}_{t|t-1} =& \Exp_{\mathcal{N}(x_{t-1}; \hat{x}_{t-1|t-1}, P_{t-1|t-1}) }  \left\{f(x_{t-1}) \right\},
\\
P_{t|t-1} =& \Exp_{\mathcal{N}(x_{t-1}; \hat{x}_{t-1|t-1}, P_{t-1|t-1})}  \left\{f(x_{t-1}) f^{\top}(x_{t-1}) \right\} \\
&+ \Var \left\{ \xi_{t-1} \right\}
- \hat{x}_{t|t-1} \hat{x}_{t|t-1}^{\top}.   
\end{aligned}    
\end{equation}
By leveraging variable substitution and the independence between process noise and state, we find that the optimal solution in \eqref{eq.BF prediction solution 1} essentially requires finding the expectation of a nonlinear function with respect to a Gaussian distribution.  
The subsequent example explains this momoment-matching operation in terms of the canonical KF:
\begin{example}[Prediction step of Kalman filter]\label{example.KF prediction}
For linear Gaussian systems $x_t = A x_{t-1} + \xi_{t-1}$ where the process noise satisfies $\xi_{t-1} \sim \mathcal{N}(\xi_{t-1}; 0, Q)$,  if the posterior distribution at time $t-1$ is given by $x_{t-1} \sim \mathcal{N}(x_{t-1}; \hat{x}_{t-1|t-1}, P_{t-1|t-1})$, the prediction step using \eqref{eq.BF prediction solution 1} yields the following stationary points:
\begin{equation}\nonumber
\begin{aligned}
\hat{x}_{t|t-1} =& \Exp_{\mathcal{N}(x_{t-1}; \hat{x}_{t-1|t-1}, P_{t-1|t-1}) }  \left\{A x_{t-1} \right\}
\\
=& A \Exp_{\mathcal{N}(x_{t-1}; \hat{x}_{t-1|t-1}, P_{t-1|t-1}) }  \left\{ x_{t-1} \right\}
\\
=& A \hat{x}_{t-1|t-1},
\\
P_{t|t-1} =& \Exp_{\mathcal{N}(x_{t-1}; \hat{x}_{t-1|t-1}, P_{t-1|t-1})}  \left\{A x_{t-1} x_{t-1}^{\top} A^{\top} \right\} \\
&+ \Var \left\{ \xi_{t-1} \right\}
- A \hat{x}_{t-1|t-1} \hat{x}_{t-1|t-1}^{\top} A^{\top}
\\
=& A P_{t-1|t-1} A^{\top} + Q. 
\end{aligned}
\end{equation}
Note that the prior covariance matrix satisfies $P_{t-1|t-1} = \Exp_{\mathcal{N}(x_{t-1}; \hat{x}_{t-1|t-1}, P_{t-1|t-1})}  \left\{x_{t-1} x_{t-1}^{\top} \right\} - \hat{x}_{t-1|t-1} \hat{x}_{t-1|t-1}^{\top}$.
\end{example}
In linear systems as shown in Example \ref{example.KF prediction},  the expectation  operator in \eqref{eq.BF prediction solution 1} is allowed to be interchanged with the affine function. However, for nonlinear functions \(f\), computing the expectation requires numerical methods.
As discussed in Section \ref{sec.1}, moment-matching KF methods, such as UKF, GHKF, and CKF, approximate this expectation using techniques like the unscented transform, Gauss–Hermite quadrature, and spherical cubature. These methods provide optimal Gaussian approximations for the prediction step by numerically solving the expectation.

\subsection{Structure of Stationary Points of \texorpdfstring{$J(\hat{x}_t, P_t)$}{J(x\_t, P\_t)}}\label{sec.stationary points of J}
Compared to the simple form of the stationary points of \(L(\hat{x}_t, P_t)\), the stationary points of \(J(\hat{x}_t, P_t)\) are relatively more complicated. For simplicity in notation, we define the negation of the log-likelihood as \(\ell(x_t, y_t)\), where \(\ell(x_t, y_t) = -\log{p(y_t|x_t)}\), and refer to \(\ell(x_t, y_t)\) as the measurement-dependent loss. Using the analytical form of the KL divergence for two Gaussian distributions, the update cost can be formulated as
\begin{equation}\label{eq.explict form of J}
\begin{aligned}
&J(\hat{x}_t, P_t) 
\\
=& \Exp_{\mathcal{N}(x_t; \hat{x}_t, P_t)} \left\{\ell(x_t, y_t)\right\}\\
&+ D_{\text{KL}}\left(\mathcal{N}(x_t; \hat{x}_t, P_t) || \mathcal{N}(x_t;\hat{x}_{t|t-1}, P_{t|t-1})\right)
\\
=& \Exp_{\mathcal{N}(x_t; \hat{x}_t, P_t)} \left\{\ell(x_t, y_t)\right\} 
\\
&+ \frac{1}{2} \left(\hat{x}_{t|t-1}-\hat{x}_t \right)^{\top} P_{t|t-1}^{-1}\left(\hat{x}_{t|t-1}-\hat{x}_t \right)
\\
&+ \frac{1}{2} \mathrm{Tr}\left(P_{t|t-1}^{-1} P_t \right)
-\frac{1}{2} \log \frac{\left|P_t\right|}{\left|P_{t|t-1}\right|}-\frac{1}{2}n .
\end{aligned}
\end{equation}
To find the stationary points of \(J(\hat{x}_t, P_t)\), we need to calculate the partial derivatives with respect to \(\hat{x}_t\) and \(P_t\). The following lemma is helpful in simplifying these partial derivative calculations:
\begin{lemma}[Gradient of expectation under Gaussian distribution]\label{lemma.Gaussian expectation}
Assuming that $f: \mathbb{R}^n \to \mathbb{R}$ is twice differentiable, we have the following results:

(i). The gradient of $\Exp_{\mathcal{N}(x; \mu, \Sigma)} \left\{ f(x) \right\}$ w.r.t. the mean $\mu$ satisfies
\begin{equation}\label{eq.Bonnet’s Theorem}
\begin{aligned}
\frac{\partial}{\partial \mu} \Exp_{\mathcal{N}(x; \mu, \Sigma)} \left\{ f(x) \right\} &= \Exp_{\mathcal{N}(x; \mu, \Sigma)} \left\{ \frac{\partial}{\partial x} f(x) \right\} 
\\
&= \Sigma^{-1} \Exp_{\mathcal{N}(x; \mu, \Sigma)} \left\{ (x-\mu) f(x) \right\}.
\end{aligned}
\end{equation}

(ii). The Hessian matrix of  $\Exp_{\mathcal{N}(x; \mu, \Sigma)} \left\{ f(x) \right\}$  w.r.t. the mean $\mu$ satisfies
\begin{equation}\label{eq.Price’s Theorem}
\begin{aligned}
&\frac{\partial^2}{\partial \mu^2 } \Exp_{\mathcal{N}(x; \mu, \Sigma)} \left\{ f(x) \right\} 
\\
=& \Exp_{\mathcal{N}(x; \mu, \Sigma)} \left\{ \frac{\partial^2}{\partial x^2} f(x) \right\}\\
=& -2 \Sigma^{-1} \left( \frac{\partial}{\partial \Sigma^{-1} }\Exp_{\mathcal{N}(x; \mu, \Sigma)} \left\{ f(x) \right\} \right)  \Sigma^{-1}
\\
=& \Sigma^{-1} \Exp_{\mathcal{N}(x; \mu, \Sigma)} \left\{ (x-\mu) (x-\mu)^{\top} f(x) \right\} \Sigma^{-1}
\\
&-\Sigma^{-1} \Exp_{\mathcal{N}(x; \mu, \Sigma)}\left\{f(x) \right\}.
\end{aligned}
\end{equation}
\end{lemma}
\begin{proof}
The results presented are all related to Stein's lemma \cite{stein1981estimation}. Result (i) is known as Bonnet’s Theorem \cite{bonnet1964transformations} and result (ii) is referred to as Price’s Theorem \cite{price1958useful}. Detailed discussions about these results can be found in this technical report \cite{lin2019stein}. \end{proof}
Before studying the structure of the stationary points, we have the following assumption:
\begin{assumption}\label{assump.twice differentiable}
The measurement-dependent loss function $\ell(x_t, y_t)$ is twice differentiable with respect to the state $x_t$.    
\end{assumption}
Based on Assumption \ref{assump.twice differentiable} and the result of Lemma \ref{lemma.Gaussian expectation}, we have the partial derivative of $J(\hat{x}_t, P_t)$ with respect to $\hat{x}_t$ and $P_t^{-1}$:
\begin{subequations}\label{eq.derivative}
\begin{align}
&\frac{\partial J(\hat{x}_t, P_t)}{\partial \hat{x}_t} \nonumber
\\
=& \Exp_{\mathcal{N}(x_t; \hat{x}_t, P_t)} \left\{ \frac{\partial \ell(x_t, y_t)}{\partial x_t} \right\} + P_{t|t-1}^{-1}(\hat{x}_t - \hat{x}_{t|t-1}), \label{eq.derivative x}
\\
&\frac{\partial J(\hat{x}_t, P_t)}{\partial P_t^{-1}} \nonumber
\\
=& -\frac{1}{2} P_t \cdot \Exp_{\mathcal{N}(x_t; \hat{x}_t, P_t)} \left\{ \frac{\partial^2 \ell(x_t, y_t)}{\partial x_t^2} \right\} \cdot P_t \nonumber
\\
&- \frac{1}{2}P_t P^{-1}_{t|t-1}P_t + \frac{1}{2}P_t.
\label{eq.derivative P}
\end{align}
\end{subequations}
To find extrema, we could attempt to set the first derivatives to zero $$\frac{\partial{J}(\hat{x}_{t|t}, P_{t|t})}{\partial{\hat{x}_{t|t}}} = \frac{\partial{J}(\hat{x}_{t|t}, P_{t|t})}{\partial{P_{t|t}^{-1}}} = 0,$$ we have
\begin{subequations}\label{eq.Gaussian BF}
\begin{align}
\hat{x}_{t|t} &=  \hat{x}_{t|t-1} - P_{t|t-1}\Exp_{\mathcal{N}(x_t; \hat{x}_{t|t}, P_{t|t})} \left\{ \frac{\partial \ell(x_t, y_t)}{\partial x_t} \right\}, \label{eq.Gaussian BF x}
\\
P_{t|t}^{-1} &= P_{t|t-1}^{-1} + \Exp_{\mathcal{N}(x_t; \hat{x}_{t|t}, P_{t|t})} \left\{ \frac{\partial^2 \ell(x_t, y_t)}{\partial x_t^2} \right\}. \label{eq.Gaussian BF P}
\end{align}    
\end{subequations}
As shown in \eqref{eq.Gaussian BF}, the first-order condition is generally not possible to isolate for $\hat{x}_{t|t}$ and $P_{t|t}$. An exception for it is the well-known KF, where the expectation in \eqref{eq.Gaussian BF} is a constant value, as shown in the subsequent example:
\begin{example}[Update step of Kalman filter]\label{example.KF update}
For linear Gaussian systems with output probability $p(y_t|x_t) = \mathcal{N}(y_t; Cx_t, R)$, the optimal Gaussian approximation of the posterior mean in  \eqref{eq.Gaussian BF x} can be written as 
\begin{equation}\label{eq.KF x update}
\begin{aligned}
&\hat{x}_{t|t} 
\\
=&  \hat{x}_{t|t-1} + P_{t|t-1}\Exp_{\mathcal{N}(x_t; \hat{x}_{t|t}, P_{t|t})} \left\{\frac{\partial}{\partial{x_t}} \log\left\{\mathcal{N}(y_t;Cx_t, R)\right\} \right\}
\\
=&
\hat{x}_{t|t-1} - P_{t|t-1}\Exp_{\mathcal{N}(x_t; \hat{x}_{t|t}, P_{t|t})} \left\{\frac{\partial}{\partial{x_t}} \left\{\frac{1}{2} \|y_t - Cx_t\|^2_{R^{-1}} \right\} \right\}
\\
=&
\hat{x}_{t|t-1} + P_{t|t-1}\Exp_{\mathcal{N}(x_t; \hat{x}_{t|t}, P_{t|t})} \left\{C^{\top}R^{-1}(y_t-Cx_t) \right\}
\\
=&
\hat{x}_{t|t-1} + P_{t|t-1}C^{\top}R^{-1}(y_t-C\hat{x}_{t|t}).
\end{aligned}
\end{equation}
From \eqref{eq.KF x update}, we can obtain
\begin{equation}\label{eq.KF x update2}
\hat{x}_{t|t} = \hat{x}_{t|t-1} + K_t(y_t-C\hat{x}_{t|t-1}), 
\end{equation}
with $K_t$ being the gain matrix of KF defined as
$K_t \triangleq P_{t|t-1}C^{\top}(CP_{t|t-1}C^{\top} + R)^{-1}$. Similarly, the optimal Gaussian approximation of the posterior covariance \eqref{eq.Gaussian BF P} can be written as
\begin{equation}\label{eq.KF P update}
\begin{aligned}
&P_{t|t} \\
=& \Big(P_{t|t-1}^{-1} - \Exp_{\mathcal{N}(x_t; \hat{x}_{t|t}, P_{t|t})} \big\{\frac{\partial^2}{\partial{x_t^{2}}} \log\left\{\mathcal{N}(y_t;Cx_t, R)\right\} \big\}\Big)^{-1}
\\
=& 
\left(P_{t|t-1}^{-1} + C^{\top}R^{-1}C \right)^{-1}
\\
=&
P_{t|t-1} - P_{t|t-1}C^{\top}(R+CP_{t|t-1}C^{\top})^{-1}CP_{t|t-1}.
\end{aligned}
\end{equation}
 Note that \eqref{eq.KF x update2} and \eqref{eq.KF P update} consist of the update step of the canonical KF.
\end{example}
As shown in Example \ref{example.KF update}, for the special case of a linear Gaussian system, the extreme conditions are specified by two decoupled equations whose analytical solutions correspond to the analytical form of the KF update. Specifically, the expectation on the right-hand side of \eqref{eq.Gaussian BF x} depends solely on the posterior mean, while the expectation in \eqref{eq.Gaussian BF P} is determined entirely by the system, meaning it is unrelated to either the mean or covariance. However, for a general nonlinear or non-Gaussian system, the expectations in \eqref{eq.Gaussian BF} depend on both the posterior mean and covariance, rendering the stationary points analytically intractable.

\begin{remark}
As discussed in Section \ref{sec.2}, existing Gaussian filters rely on approximation techniques that solve \eqref{eq.Gaussian BF} by linearizing the measurement model and then performing the KF updates as \eqref{eq.KF x update2} and \eqref{eq.KF P update}. We contend that this type of approximation technique inevitably introduces linearization errors. For example, the Taylor series expansion technique unavoidably results in higher-order error terms. Moreover, these approximation techniques are applicable only for Gaussian noises \cite{sarkka2023bayesian}. Therefore, there is an urgent need to develop a new method that can directly solve the Gaussian approximation for Bayesian filtering update to avoid linearization errors.    
\end{remark}

\section{Natural Gradient Gasussian Approximation}\label{sec.NANO}
In the previous section, we observed that existing moment-matching KF methods already provide exact numerical solutions for the optimal Gaussian approximation in the prediction step. However, the update step is not sufficiently resolved in current Gaussian filters, as these methods rely on linearization to approximate the stationary point, which introduce  linearization errors. In this section, we aim to tackle \textbf{Q2}. Specifically, we will design an algorithm to solve \eqref{eq.Gaussian optimization update}.

By examining the structure of the extreme conditons defined by \eqref{eq.Gaussian BF}, we find that it is challenging to directly obtain an analytical form of the stationary point because it is typically impossible to isolate the updates for the mean and covariance in \eqref{eq.Gaussian BF}. Therefore, a more practical approach to finding the optimal solution is to directly minimize the update cost  \(J(\hat{x}_t, P_t)\). To find the steepest descent in optimizing the parameters of Gaussian distributions \cite{amari1998natural,martens2020new}, we derive a natural gradient iteration for finding the optimal Gaussian approximation.

For simplicity, we stack the Gaussian parameters into a single column vector $v$ and calculate the derivative with respect to it: 
\begin{equation}\label{eq.stack definition}
\begin{aligned}
v = \begin{bmatrix}
\hat{x}_t \\ \rvec(P_t^{-1})
\end{bmatrix}, \; \frac{\partial}{\partial{v}} J(\hat{x}_t, P_t) &= \begin{bmatrix}
\frac{\partial}{\partial \hat{x}_t} J(\hat{x}_t, P_t)  \\ \rvec\left(\frac{\partial}{\partial P_t^{-1}} J(\hat{x}_t, P_t)\right)
\end{bmatrix}.
\end{aligned}
\end{equation}
Here, we consider the inverse of the covariance matrix instead of its original form. This consideration is inspired by the structure of the information filter \cite{fraser1967new,anderson2005optimal}, an equivalent form of the KF, where the inverse of the covariance matrix is employed instead of the covariance matrix itself. This is because the inverse can potentially simplify the mathematical expression of the update step in Bayesian filtering \cite{fraser1967new,anderson2005optimal}.
Additionally, to easily represent the iteration, we define 
\begin{equation}\nonumber
\delta v \triangleq \begin{bmatrix}
\delta \hat{x}_t \\ \rvec(\delta P_t^{-1})
\end{bmatrix} = 
\begin{bmatrix}
\hat{x}_t^{(i+1)} - \hat{x}_t^{(i)} \\
\mathrm{vec} \left(\left( P_t^{-1} \right)^{(i+1)} - \left(P_t^{-1} \right)^{(i)}\right)
\end{bmatrix}
,       
\end{equation}
where $i$ is the iteration index. Under this notation, the natural gradient parameter update can be defined as
\begin{equation}\label{eq.natural gradient update}
\delta v = - \left[ \mathcal{F}_v^{-1} \frac{\partial}{\partial v}  J(\hat{x}_t, P_t) \right]_{v=v^{(i)}},  
\end{equation}
where $\mathcal{F}_v$ is the fisher information matrix associated with the Gaussian distribution $\mathcal{N}(x_t; \hat{x}_t, P_t)$ and $v^{(i)}$ represents the value of $v$ in the $i$-th iteration. The next proposition provides the formulation of Fisher information matrix:
\begin{proposition}\label{prop.fisher information matrix}
The inverse of the Fisher information matrix $\mathcal{F}^{-1}_v$ associated with $\mathcal{N}(x_t; \hat{x}_t, P_t)$ is 
\begin{equation}\label{eq.fisher matrix inverse}
\begin{aligned}
\mathcal{F}^{-1}_v = 
\begin{bmatrix}
P_t & 0 \\
0 & 2 (P_t^{-1} \otimes P_t^{-1})
\end{bmatrix}, 
\end{aligned}   
\end{equation}
where $\otimes$ is the kronecker product.
\end{proposition}
The proof of this proposition can be found in Section 2.2 of \cite{barfoot2020multivariate}.
Combining \eqref{eq.stack definition} and \eqref{eq.fisher matrix inverse} with \eqref{eq.natural gradient update}, we have
\begin{equation}\nonumber
\begin{aligned}
\delta \hat{x}_t &= - \left[ P_t \frac{\partial}{\partial \hat{x}_t} J(\hat{x}_t, P_t) \right]_{v=v^{(i)}}, 
\\
\rvec \left(\delta P_t^{-1} \right) &= -2 \left[ \left(P_t^{-1} \otimes P_t^{-1} \right) \rvec\left(\frac{\partial}{\partial P_t^{-1}} J(\hat{x}_t, P_t)\right) \right]_{v=v^{(i)}}.
\end{aligned}    
\end{equation}
Transforming this into matrix form, we can derive the following iterative updates:
\begin{equation}\label{eq.iterative updates}
\begin{aligned}
\left( P_t^{-1} \right)^{(i+1)} =& \left( P_t^{-1} \right)^{(i)}
\\
& -2 \left( P_t^{-1} \right)^{(i)}  
\left. \frac{\partial}{\partial P_t^{-1}} J(\hat{x}_t, P_t) \right|_{v^{(i)}} 
\left( P_t^{-1} \right)^{(i)}, 
\\
\hat{x}^{(i+1)}_t =& \hat{x}^{(i)}_t - P_t^{(i+1)} 
\left.
\frac{\partial}{\partial \hat{x}_t} J(\hat{x}_t, P_t) \right|_{v^{(i)}}. 
\end{aligned}
\end{equation}
Combining the iterative updates \eqref{eq.iterative updates} with the formulation of partial derivative, we have
\begin{equation}\label{eq.natural gradient update 1}
\begin{aligned}
\left( P_t^{-1} \right)^{(i+1)} =& P_{t|t-1}^{-1}  + \Exp_{\mathcal{N}(x_t; \hat{x}_t^{(i)}, P_t^{(i)})} \left\{ \frac{\partial^2 \ell(x_t, y_t)}{\partial x_t^2} \right\},\\
\hat{x}^{(i+1)}_t =& \hat{x}^{(i)}_t - P_t^{(i+1)} 
\Exp_{\mathcal{N}(x_t; \hat{x}_t^{(i)}, P_t^{(i)})} \left\{ \frac{\partial \ell(x_t, y_t)}{\partial x_t} \right\} \\
&-
P_t^{(i+1)} P_{t|t-1}^{-1}(\hat{x}_t^{(i)} - \hat{x}_{t|t-1}).
\end{aligned}  
\end{equation}
One practical issue when performing \eqref{eq.natural gradient update 1} is that the derivatives of $\ell(x_t, y_t)$ can be hard to compute. To avoid the need to compute derivatives of
the measurement-dependent loss, we can once again apply Lemma \ref{lemma.Gaussian expectation} to acquire the derivative-free formulation. By applying \eqref{eq.Bonnet’s Theorem} and \eqref{eq.Price’s Theorem}, we have
\begin{equation}\label{eq.derivative-free formulation}
\begin{aligned}
&\Exp_{\mathcal{N}(x_t; \hat{x}_t^{(i)}, P_t^{(i)})} \left\{ \frac{\partial \ell(x_t, y_t)}{\partial x_t} \right\} 
\\
=& \left( P_t^{-1} \right)^{(i)} 
\Exp_{\mathcal{N}(x_t; \hat{x}_t^{(i)}, P_t^{(i)})} \left\{ \left(x_t - x_t^{(i)} \right)\ell(x_t, y_t) \right\},
\\
&\Exp_{\mathcal{N}(x_t; \hat{x}_t^{(i)}, P_t^{(i)})} \left\{ \frac{\partial^2 \ell(x_t, y_t)}{\partial x_t^2} \right\}
\\
=& \left( P_t^{-1} \right)^{(i)} \Exp_{\mathcal{N}(x_t; \hat{x}_t^{(i)}, P_t^{(i)})} \big\{ (x_t - \hat{x}_t^{(i)} ) (x_t - \hat{x}_t^{(i)})^{\top} \ell(x_t, y_t) \big\}  \left( P_t^{-1} \right)^{(i)}
\\
&- \left( P_t^{-1} \right)^{(i)} \Exp_{\mathcal{N}(x_t; \hat{x}_t^{(i)}, P_t^{(i)})} \left\{
\ell(x_t, y_t) \right\}.
\end{aligned}
\end{equation}
With the result in \eqref{eq.derivative-free formulation}, we have the derivative-free update scheme shown in \eqref{eq.derivative-free update}.
\begin{figure*}[t]
\begin{equation}\label{eq.derivative-free update}
\begin{aligned}
\left( P_t^{-1} \right)^{(i+1)} =& P_{t|t-1}^{-1} 
+ 
\left( P_t^{-1} \right)^{(i)} \cdot \Exp_{\mathcal{N}(x_t; \hat{x}_t^{(i)}, P_t^{(i)})} \Big\{ (x_t - \hat{x}_t^{(i)} ) (x_t - \hat{x}_t^{(i)} )^{\top} \ell(x_t, y_t) \Big\} \cdot \left( P_t^{-1} \right)^{(i)} - \left( P_t^{-1} \right)^{(i)} \Exp_{\mathcal{N}(x_t; \hat{x}_t^{(i)}, P_t^{(i)})} \left\{
\ell(x_t, y_t) \right\},
\\
\hat{x}^{(i+1)}_t =& \hat{x}^{(i)}_t - P_t^{(i+1)} \cdot \left( P_t^{-1} \right)^{(i)} \cdot
\Exp_{\mathcal{N}(x_t; \hat{x}_t^{(i)}, P_t^{(i)})} \left\{ \left(x_t - x_t^{(i)} \right)\ell(x_t, y_t) \right\} -
P_t^{(i+1)} P_{t|t-1}^{-1}(\hat{x}_t^{(i)} - \hat{x}_{t|t-1}).
\end{aligned}    
\end{equation}
\end{figure*}
This update scheme is still practically intractable for two reasons. First, the expectations generally do not have analytical forms. Second, the update scheme in \eqref{eq.derivative-free update} generally cannot guarantee that the covariance matrix will be positive definite. 

To address the first issue, we could use well-established numerical integration methods, such as the unscented transform \cite{julier1995new}, Gauss–Hermite integration \cite{golub1969calculation}, or spherical cubature integration \cite{arasaratnam2009cubature} to approximate the expectation calculations.
For the second issue, one possible and efficient solution is to provide a sufficiently good initialization. For example, we could solve the maximum a posterior estimation problem:
\begin{equation}\label{eq.MAP initialziation}
\begin{aligned}
\hat{x}_{t|t}^{\mathrm{MAP}} &= \argmax_{x_t} \left\{
\mathcal{N}(x_t; \hat{x}_{t|t-1}, P_{t|t-1})
 \cdot \exp\{-\ell(x_t, y_t)\} \right\}
\\
& = \argmax_{x_t} \left\{
\log{\mathcal{N}(x_t; \hat{x}_{t|t-1}, P_{t|t-1})}
 - \ell(x_t, y_t) \right\}, 
\end{aligned}
\end{equation}
and use Laplace's approximation \cite{kass1991laplace} to construct the initial mean and covariance for the iteration in \eqref{eq.derivative-free update}, as shown in the subsequent equation:
{\small 
\begin{equation}\label{eq.initialization}
\begin{aligned}
\hat{x}_{t}^{(0)} =& \hat{x}_{t|t}^{\mathrm{MAP}},
\\
\left(P_{t}^{-1}\right)^{(0)} = &  \left. \frac{\partial^2 \left\{-\log{\mathcal{N}(x_t; \hat{x}_{t|t-1}, P_{t|t-1})}
 + \ell(x_t, y_t)\right\}}{\partial x_t^{2}}\right|_{x_t = \hat{x}_{t|t}^{\mathrm{MAP}}}.
\end{aligned}
\end{equation}}    
This method works quite well in most scenarios. Besides this initialization trick, other methods such as Cholesky decomposition \cite{salimbeni2018natural} or square-root parameterization \cite{glasmachers2010exponential} can also be leveraged to ensure the positive definiteness of the covariance matrix.

Another important consideration for the iterative scheme is the stopping criterion. As suggested by \cite{garcia2015posterior}, we use the KL divergence between two consecutive Gaussian distributions, specifically $\mathcal{N}^{(i)}= \mathcal{N}(x_t; \hat{x}_{t|t}^{(i)}, P_{t|t}^{(i)})$ and $\mathcal{N}^{(i+1)} = \mathcal{N}(x_t; \hat{x}_{t|t}^{(i+1)}, P_{t|t}^{(i+1)})$ to determine when to stop the iteration: \begin{equation}\label{eq.stopping criterion} D_{\mathrm{KL}}(\mathcal{N}^{(i)} \| \mathcal{N}^{(i+1)}) < \gamma, \end{equation} where $\gamma$ is a predefined threshold. This approach is more effective compared to using $D_{\mathrm{KL}}(\mathcal{N}^{(i+1)} \| \mathcal{N}^{(i)}) < \gamma$ because $\mathcal{N}^{(i+1)}$ is generally more concentrated. By using $D_{\mathrm{KL}}(\mathcal{N}^{(i)} \| \mathcal{N}^{(i+1)})$, we ensure that the criterion remains sensitive to convergence while avoiding premature termination of the algorithm.

Recall that in Section \ref{sec.stationary points of L}, we proved that the prediction step of the moment-matching KF algorithms essentially follows the optimal solution of Gaussian filtering. By combining this step with our natural gradient descent update step \eqref{eq.derivative-free update}, we developed a new iterative filter. To emphasize that natural gradient descent is our key contribution, we call it the \textbf{N}atural gr\textbf{A}dient Gaussia\textbf{N} appr\textbf{O}ximation filter, or NANO filter for short. The pseudocode of the NANO filter is summarized in Algorithm~\ref{alg.1}. Note that all the expectation computations appearing in Algorithm~\ref{alg.1} are suggested to use the efficient unscented transform \cite{julier2004unscented}.

\begin{algorithm}[t]
\caption{NANO Filter}\label{alg.1}
\begin{algorithmic} 
    \State \textbf{Input:} Stopping thresohold $\gamma$
    \State \textbf{Initialization:} State estimate $\hat{x}_{0|0}$ and covariance $P_{0|0}$
    \For{each time step $t$}
        \State \textbf{Predict:}
        \State Calculate predicted state mean $\hat{x}_{t|t-1}$ and covariance $P_{t|t-1}$ using \eqref{eq.BF prediction solution 1} 
        \State \textbf{Update:}
        \State Obtain the noisy measurement $y_t$ 
        \State  Initialize the state estimate $\hat{x}_{t}^{(0)}$ and covariance $P_{t}^{(0)}$ using \eqref{eq.initialization}
        \For{each iteration number $i$} \If{\eqref{eq.stopping criterion} is not satisfied}
        \State {Update state estimate and covariance using  \eqref{eq.derivative-free update}}
        \EndIf
        \EndFor
        \State 
        $
        \hat{x}_{t|t} = \hat{x}_t^{(i)}, P_{t|t} = P_t^{(i)}
        $
    \EndFor
\end{algorithmic}
\end{algorithm}

\section{Theoretical Analysis}\label{sec.theoretical analysis}
In this section, we will answer \textbf{Q3}, i.e., we will provide the convergence and stability analysis for NANO filter.
\subsection{Convergence Analysis}
After deriving the update scheme of natural gradient Gaussian filtering in  \eqref{eq.natural gradient update 1}, a key question arises: does it converge to the optimal solution of  \eqref{eq.Gaussian optimization update}? The following theorem confirms the local convergence of our update scheme:
\begin{theorem}\label{theorem.local convergence}
Consider the Taylor series expansion that is second order in $\delta \hat{x}_t$ and first order in $\delta P_t^{-1}$. Under this approximation, the iterative update in \eqref{eq.natural gradient update 1} guarantees convergence, i.e.,
\begin{equation}\label{eq.Taylor series expansion}
\begin{aligned}
J_t^{(i+1)} \approx&\ J_t^{(i)} +  \left. \frac{\partial{J}}{\partial{\hat{x}_t^\top}} \right|_{v^{(i)}} \delta \hat{x}_t 
+ \frac{1}{2} \left(\delta \hat{x}_t\right)^\top \left( \left. \frac{\partial^2{J}}{\partial{\hat{x}_t^2}} \right|_{v^{(i)}} \right) \delta \hat{x}_t \\
&+ \mathrm{Tr} \left( \left. \frac{\partial J}{\partial P_t^{-1}} \right|_{v^{(i)}}  \delta P_t^{-1} \right) \\
\leq&\ J_t^{(i)},
\end{aligned}    
\end{equation}
where $J_t^{(i)}$ is the update cost at the $i$-th iteration, defined as $J_t^{(i)} \triangleq J(\hat{x}_t^{(i)}, P_t^{(i)})$. Moreover, equality in \eqref{eq.Taylor series expansion} holds if and only if $\delta \hat{x}_t = 0$ and $\delta P_t^{-1} = 0$.
\end{theorem}
\begin{proof}
This proof is inspired by Section 6.2.2 in \cite{barfoot2024state}. According to \eqref{eq.Price’s Theorem} and \eqref{eq.natural gradient update 1}, the second-order derivative of $J(\hat{x}_t, P_t)$ with respect to $\hat{x}_t$ satisfies
\begin{equation}\label{eq.2J x2}
\begin{aligned}
\left. \frac{\partial^2 J}{\partial \hat{x}_t^{2}} \right|_{v^{(i)}}
=&
\left. \frac{\partial^2}{\partial \hat{x}_t^2 } \Exp_{\mathcal{N}(x_t; \hat{x}_t, P_t)} \left\{\ell(x_t, y_t)\right\} \right|_{v^{(i)}} + P_{t|t-1}^{-1}
\\
=& \Exp_{\mathcal{N}(x_t; \hat{x}_t^{(i)}, P_t^{(i)})} \left\{\frac{\partial^2 \ell(x_t, y_t)}{\partial {x}_t^2 }\right\} + P_{t|t-1}^{-1}
\\
=& \left(P_t^{-1}\right)^{(i+1)}.
\end{aligned}    
\end{equation}
According to \eqref{eq.iterative updates}, we have
\begin{equation}\label{eq.delta x}
\begin{aligned}
\left. \frac{\partial J}{\partial \hat{x}_t}  \right|_{v^{(i)}} &=  - \left( P_t^{-1} \right)^{(i+1)} \delta \hat{x}_t,
\\
\left. \frac{\partial J}{\partial P_t^{-1}}  \right|_{v^{(i)}} &=  - \frac{1}{2} \left( P_t \right)^{(i)} \delta P_t^{-1} \left( P_t \right)^{(i)}.
\end{aligned}
\end{equation}
Combined with \eqref{eq.2J x2} and \eqref{eq.delta x}, the Tayor-series expansion can be expressed as
\begin{equation}\nonumber
\begin{aligned}
& J_t^{(i+1)} - J_t^{(i)} \\
\approx & -\frac{1}{2} \delta \hat{x}_t^{\top
} \left( P_t^{-1} \right)^{(i+1)} \delta \hat{x}_t -\frac{1}{2} \mathrm{Tr} \left(P_t^{(i)} \delta P_t^{-1} P_t^{(i)} \delta P_t^{-1}\right)
\\
=& -\frac{1}{2} \delta \hat{x}_t^{\top
} \left( P_t^{-1} \right)^{(i+1)} \delta \hat{x}_t 
\\
&-\frac{1}{2} \mathrm{vec}(\delta P_t^{-1})^{\top} \left( P_t^{(i)} \otimes P_t^{(i)} \right) \mathrm{vec}(\delta P_t^{-1}) 
\\
\leq 
&0.
\end{aligned}
\end{equation}\end{proof}
This theorem indicates that natural gradient descent iteration in the update step of the NANO filter provides a guarantee of local convergence. This guarantee is achieved by approximating the objective function with second-order accuracy around the mean and first-order accuracy around the inverse of the covariance matrix.
The key idea of proof is to show that the difference in update costs between consecutive iterations can be expressed as a semi-negative definite quadratic form. Technically, achieving this semi-negative definiteness relies on the use of the Fisher information matrix, which corrects the gradient direction to provide the steepest descent on the Gaussian manifold. This adjustment makes the gradient ``natural" because it aligns with the geometry of the Riemannian space of Gaussian parameters.

For linear Gaussian systems, the NANO filter achieves the optimal solution of  \eqref{eq.Gaussian optimization update} in a single iteration, as stated in the following corollary:

\begin{corollary}\label{corollary.convergence}
For linear Gaussian systems in Example \ref{example.KF update}, the update rule given by \eqref{eq.natural gradient update 1} converges to the optimal solution of \eqref{eq.Gaussian optimization update} within one iteration. In other words, a single iteration of NANO filter is equivalent to KF. \end{corollary}
\begin{proof}
For linear Gaussian systems with output probability $p(y_t|x_t) = \mathcal{N}(y_t; Cx_t, R)$, the one-step iteration of the covariance matrix in \eqref{eq.natural gradient update 1} is given by: 
\begin{equation}\label{eq.1 step iteration}
\begin{aligned}
\left( P_t^{-1} \right)^{(1)} =& P_{t|t-1}^{-1}  + \Exp_{\mathcal{N}(x_t; \hat{x}_t^{(0)}, P_t^{(0)})} \left\{ \frac{\partial^2 \ell(x_t, y_t)}{\partial x_t^2} \right\},
\\
=& P_{t|t-1}^{-1} + C^{\top} R^{-1} C.
\end{aligned}  
\end{equation}
Comparing \eqref{eq.1 step iteration} with KF covariance update equation \eqref{eq.KF P update}, we observe that $P_t^{(1)} = P_{t|t}$. Then, refering to \eqref{eq.KF x update} and using \eqref{eq.1 step iteration}, the one-step iteration of the mean vector in \eqref{eq.natural gradient update 1} can be expressed as:
\begin{equation}\label{eq.1 step iteration mean}
\begin{aligned}
\hat{x}^{(1)}_t =& \hat{x}^{(0)}_t - P_t^{(1)} 
\Exp_{\mathcal{N}(x_t; \hat{x}_t^{(0)}, P_t^{(0)})} \left\{ \frac{\partial \ell(x_t, y_t)}{\partial x_t} \right\} \\
&-
P_t^{(1)} P_{t|t-1}^{-1}\left(\hat{x}_t^{(0)} - \hat{x}_{t|t-1}\right)
\\
=& \hat{x}^{(0)}_t + P_{t|t} C^{\top} R^{-1} (y_t - C\hat{x}_{t}^{(0)})
\\
&-P_{t|t} \left(P_{t|t}^{-1} - C^{\top}R^{-1}C  \right) \left(\hat{x}_t^{(0)} - \hat{x}_{t|t-1}\right)
\\
=& \hat{x}_{t|t-1} + P_{t|t} C^{\top} R^{-1} \left( 
y_t - C\hat{x}_{t|t-1}
\right). 
\end{aligned}  
\end{equation}
Furthermore, noting that:
\begin{equation}\nonumber
\begin{aligned}
&P_{t|t} C^{\top} R^{-1} \\=& 
P_{t|t-1} C^{\top} R^{-1} 
\\
&- P_{t|t-1} C^{\top} (R+CP_{t|t-1}C^{\top})^{-1}C P_{t|t-1} C^{\top}R^{-1}
\\
=& P_{t|t-1} C^{\top} \big(I - (R+C P_{t|t-1} C^{\top})^{-1} (C P_{t|t-1} C^{\top}+R) 
\\
&+ (R+CP_{t|t-1}C^{\top})^{-1} R  \big) R^{-1}
\\
=& P_{t|t-1} C^{\top}(R+C P_{t|t-1} C^{\top})^{-1},
\end{aligned}
\end{equation}
we recognize that \( P_{t|t} C^\top R^{-1} \) is the Kalman gain \( K_t \). Therefore, \eqref{eq.1 step iteration mean} becomes:
\begin{equation}\nonumber
\hat{x}^{(1)}_t = \hat{x}_{t|t-1} + K_t\left( 
y_t - C\hat{x}_{t|t-1}
\right),   
\end{equation}
which is the standard KF mean update equation.
Thus, the one-step iteration of the NANO filter in \eqref{eq.natural gradient update 1} yields the optimal estimate, regardless of the initialization \( \hat{x}^{(0)}_t \) and \( P^{(0)}_t \).
\end{proof}
Corollary \ref{corollary.convergence} holds regardless of the initialization of the Gaussian parameters \( \hat{x}^{(0)}_t \) and \( P^{(0)}_t \), supporting the fact that the natural gradient is the steepest descent direction in the Gaussian manifold.

\subsection{Stability Analysis}
Next, we will analyze the stability of the proposed NANO filter. Stability is the most critical property, as it ensures that the estimation error remains bounded throughout the filtering process. Our stability analysis is conducted based on the state-space model given in \eqref{eq.SSM}. 
We focus on the case where the process and measurement noise are both zero-mean Gaussian noise, satisfying $\mathbf{\xi}_{t} \sim \mathcal{N}(0, Q_{t})$ and $\mathbf{\zeta}_{t} \sim \mathcal{N}(0, R_{t})$. In this case, the measurement-dependent loss $\ell(x_t, y_t)$ is the log-likelihood loss function that satisfies
$ 
\ell(x_{t},y_{t}) = -\log p(y_{t}|x_{t}) 
= C_{0} + \frac{1}{2}(y_{t}-g(x_{t}))^{\top}R_{t}^{-1}(y_{t} - g(x_{t}))$.
Here, $C_0>0$ is a constant value irrelevant to the state and measurement. For the purpose of stability analysis, we make the following regularity assumptions on the system functions $f$ and $g$.
\begin{assumption}\label{assump.bounded derivative}
The derivatives of the functions \(f\) and \(g\) are bounded. Specifically, there exists a constant \(C > 0\) such that
\begin{equation}		\left|\frac{\partial f^{i}}{\partial x^{j}}\right|, \left|\frac{\partial g^{k}}{\partial x^{j}}\right|, \left|\frac{\partial^{2} g^{k}}{\partial x^{i}\partial x^{j}}\right|\leq C < \infty,
\end{equation} 
for all $x\in\mathbb{R}^{n}$, $1\leq i,j\leq n$, $1\leq k\leq m$. Here, \(x^i\) represents the \(i\)-th component of the state vector \(x\); similarly, \(f^i\) and \(g^k\) denote the corresponding components of the functions \(f\) and \(g\), respectively. 
\end{assumption} 

Using the standardization of the state $x_t$ under the Gaussian distribution, we can rewrite the stationary point condition for update step \eqref{eq.Gaussian BF} in the following tensor form:
\begin{subequations}\label{eq.tensor stationary point}
\begin{equation}\label{eq.tensor stationary point x_hat}
\begin{aligned} 
\hat{x}_{t|t}^{i} =& \hat{x}_{t|t-1}^{i} + (P_{t|t-1})^{ij}\int (R_{t}^{-1})_{kl} \frac{\partial g^{k}}{\partial x^{j}}(\hat{x}_{t|t}+S_{t|t}z)\\
	& \times (y_{t}^{l} - g^{l}(\hat{x}_{t|t}+S_{t|t}z))\frac{1}{(2\pi)^{\frac{n}{2}}}e^{-\frac{1}{2}\|z\|^{2}} \d z,
\end{aligned}
\end{equation}
\begin{equation}\label{eq.tensor stationary point P}
\begin{aligned} 
(P_{t|t}^{-1})_{ij} =&  (P_{t|t-1}^{-1})_{ij} \\
&+ \int  (R_{t}^{-1})_{kl}\biggl[\frac{\partial g^{k}}{\partial x^{i}}(\hat{x}_{t|t} + S_{t|t}z)\frac{\partial g^{l}}{\partial x^{j}}(\hat{x}_{t|t} + S_{t|t}z) \\
&\qquad-\frac{\partial^{2}g^{k}}{\partial x^{i}\partial x^{j}}(\hat{x}_{t|t}+S_{t|t}z) (y_{t}^{l}-g^{l}(\hat{x}_{t|t}+S_{t|t}z))\biggr]\\	&\qquad\times\frac{1}{(2\pi)^{\frac{n}{2}}}e^{-\frac{1}{2}\|z\|^{2}} \d z,
\end{aligned} 
\end{equation}
\end{subequations}
where $S_{t|t}S_{t|t}^{\top} = P_{t|t}$ is obtained by Cholesky decomposition, and we use the Einstein summation convention \cite{einstein2003meaning}. 

Next, we will substitute the measurement model, $y_{t}^{l} = g^{l}(x_{t}) + \zeta_{t}^{l}$ into the tensor form of the stationary point condition in \eqref{eq.tensor stationary point}, and consider the Taylor expansion of $g(x_t)$ at $\hat{x}_{t|t}$. For the stationary point condition of $\hat{x}_{t|t}$, let us first define the auxiliary function $h_{j}^{kl}(x)$, for $1 \leq j \leq n$ and $1 \leq k,l \leq n$:
\begin{equation}\nonumber
\begin{aligned} 
h_{j}^{kl}(x) &= \int  \frac{\partial g^{k}}{\partial x^{j}}(\hat{x}_{t|t}+S_{t|t}z)
\\
&\times(g^{l}(x) - g^{l}(\hat{x}_{t|t}+S_{t|t}z))\frac{1}{(2\pi)^{\frac{n}{2}}}e^{-\frac{1}{2}\|z\|^{2}} \d z,
\end{aligned} 
\end{equation}
then \eqref{eq.tensor stationary point x_hat} can be rewritten as
\begin{equation}\label{eq.tensor stationary point x_hat 1}
\begin{aligned} 
\hat{x}_{t|t}^{i} =& \hat{x}_{t|t-1}^{i} + (P_{t|t-1})^{ij}(R_{t}^{-1})_{kl}h_{j}^{kl}(x_{t}) 
\\
&+ (P_{t|t-1})^{ij}(R_{t}^{-1})_{kl} \Exp_{\mathcal{N}(x_t; \hat{x}_{t|t},P_{t|t})}\biggl\{\frac{\partial g^{k}}{\partial x^{j}}(x_t)\biggr\}\zeta_{t}^{l}
 \\
		 =&\hat{x}_{t|t-1}^{i} + (P_{t|t-1})^{ij}(R_{t}^{-1})_{kl}\\
		 &\times\biggl[h_{j}^{kl}(\hat{x}_{t|t}) + \frac{\partial h_{j}^{kl}}{\partial x^{q}}(x_{t}^{q}-\hat{x}_{t|t}^{q}) + \psi_{j}^{kl}(x_{t}-\hat{x}_{t|t})\biggr]\\
		&+ (P_{t|t-1})^{ij}(R_{t}^{-1})_{kl} \Exp_{\mathcal{N}(x_t; \hat{x}_{t|t},P_{t|t})}\biggl\{\frac{\partial g^{k}}{\partial x^{j}}(x_t)\biggr\} \zeta_{t}^{l},
	\end{aligned}
\end{equation}
where $$\psi_{j}^{kl}(x_{t} - \hat{x}_{t|t})\leq \kappa_{1}(\|x_{t}-\hat{x}_{t|t}\|^{2}), \forall \ 0 \leq \|x_{t}-\hat{x}_{t|t}\|<\epsilon_{1},$$ for some $\kappa_{1},\epsilon_{1} > 0$, are high-order terms in the Taylor expansion of $h^{kl}_j$.

Let us denote the estimation error as $e_{t|t} = x_{t} - \hat{x}_{t|t}$ and $e_{t|t-1} = x_{t} - \hat{x}_{t|t-1}$. Regardless of the specific methods used in the prediction step, the prediction error $e_{t|t-1}$ can be expressed as 
\begin{equation}\label{eq.taylor expansion of e_{t|t-1}}
\begin{aligned} 
e_{t|t-1}^{i} = &f^{i}(x_{t-1}) + \xi_{t}^{i}- \hat{x}_{t|t-1}^{i}
\\
= &\frac{\partial f^{i}}{\partial x^{j}}(\hat{x}_{t-1|t-1})(x_{t-1}^{j} - \hat{x}_{t-1|t-1}^{j}) 
\\
&+ \tilde{\psi}^{i}(x_{t-1}-\hat{x}_{t-1|t-1}) + \xi_{t}^{i},
\end{aligned} 
\end{equation}
where
\begin{equation}\nonumber
\begin{aligned}
\tilde{\psi}^{i}(x_{t-1} - \hat{x}_{t-1|t-1} ) &\leq \kappa_{2}(\|x_{t-1}-\hat{x}_{t-1|t-1}\|^{2}), 
\\
\forall \; 0 &\leq \|x_{t-1}-\hat{x}_{t-1|t-1}\|<\epsilon_{2},    
\end{aligned}
\end{equation}
for some $\kappa_{2},\epsilon_{2} > 0$, which are higher-order terms in the Taylor expansion of \( f^i \). Subtracting \( x_t \) from both sides of \eqref{eq.tensor stationary point x_hat 1} and applying \eqref{eq.taylor expansion of e_{t|t-1}}, we obtain
\begin{equation}\label{eq.taylor expansion of e_{t|t}}
\begin{aligned}
e_{t|t}^{i} =& \frac{\partial f^{i}}{\partial x^{j}}(\hat{x}_{t-1|t-1})e_{t-1|t-1}^{j} + \tilde{\psi}^{i}(x_{t-1}-\hat{x}_{t-1|t-1}) 
\\
&+\xi_{t}^i-(P_{t|t-1})^{ij}(R_{t}^{-1})_{kl}
\\
&\times\biggl[h_{j}^{kl}(\hat{x}_{t|t}) + \frac{\partial h_{j}^{kl}}{\partial x^{q}}(\hat{x}_{t|t})e_{t|t}^{q} + \psi_{j}^{kl}(x_{t}-\hat{x}_{t|t})\biggr]
\\
&- (P_{t|t-1})^{ij}(R_{t}^{-1})_{kl}\Exp_{\mathcal{N}(x_t;\hat{x}_{t|t},P_{t|t})}\biggl\{\frac{\partial g^{k}}{\partial x^{j}}(x_t)\biggr\}\zeta_{t}^{l}.
\end{aligned}
\end{equation}
In matrix form, \eqref{eq.taylor expansion of e_{t|t}} becomes
\begin{equation}\nonumber
e_{t|t} = F_{t-1}e_{t-1|t-1} - H_{t}e_{t|t} - \bar{h}_{t} + \Psi_{t} + \xi_{t} - G_{t}\zeta_{t},
\end{equation}
where $F_{t}$, $H_{t}$, $\bar{h}_{t}$ and $G_{t}$ are matrix- or vector-valued functions with components given by
\begin{equation}\nonumber
\begin{aligned}
 (F_{t-1})^i_j &= \frac{\partial f^{i}}{\partial x^{j}}(\hat{x}_{t-1|t-1}),
 \\
(H_{t})_{q}^{i} &= (P_{t|t-1})^{ij}(R_{t}^{-1})_{kl}\frac{\partial h_{j}^{kl}}{\partial x^{q}}(\hat{x}_{t|t}),
\\
\bar{h}_{t}^{i} &= (P_{t|t-1})^{ij}(R_{t}^{-1})_{kl}h_{j}^{kl}(\hat{x}_{t|t}),
\\
(G_{t})_{l}^{i} &= (P_{t|t-1})^{ij}(R_{t}^{-1})_{kl} \Exp_{\mathcal{N}(x_t;\hat{x}_{t|t},P_{t|t})}\biggl\{\frac{\partial g^{k}}{\partial x^{j}}(x_t)\biggr\},
\end{aligned}
\end{equation}
and 
\begin{equation}\nonumber
\begin{aligned} 
\Psi_{t}^{i} &= \tilde{\psi}^{i}(e_{t-1|t-1}) -(P_{t|t-1})^{ij}(R_{t}^{-1})_{kl}\psi_{j}^{kl}(x_{t}-\hat{x}_{t|t})
\end{aligned}
\end{equation}
are the high-order terms. Therefore, we have
\begin{equation}\label{eq.evolution of e}
\begin{aligned}
e_{t|t} = &(I+H_{t})^{-1}F_{t-1}e_{t-1|t-1} - (I+H_{t})^{-1}\bar{h}_{t} \\
&+ (I+H_{t})^{-1}\Psi_{t} + (I+H_{t})^{-1}\left(\xi_{t}-G_{t} \zeta_{t}\right).
\end{aligned} 
\end{equation}

Similarly, for the stationary point condition of $P_{t|t}$, let us define the auxiliary function $\tilde{h}_{ij}^{kl}(x)$, for $1 \leq i,j \leq n$ and $1 \leq k,l \leq n$:
\begin{equation}\nonumber
\begin{aligned} 
\tilde{h}_{ij}^{kl}(x) =& \int\biggl[\frac{\partial g^{k}}{\partial x^{i}}(\hat{x}_{t|t} + S_{t|t}z)\frac{\partial g^{l}}{\partial x^{j}}(\hat{x}_{t|t} + S_{t|t}z)\\
		&-\frac{\partial^{2}g^{k}}{\partial x^{i}\partial x^{j}}(\hat{x}_{t|t}+S_{t|t}z) (g^{l}(x)-g^{l}(\hat{x}_{t|t}+S_{t|t}z))\biggr]\\
		&\times\frac{1}{(2\pi)^{\frac{n}{2}}}e^{-\frac{1}{2}\|z\|^{2}} \d z,
	\end{aligned} 
\end{equation}
then \eqref{eq.tensor stationary point P} can be rewritten as
\begin{equation}\nonumber
\begin{aligned}
(P_{t|t}^{-1})_{ij} = &(P_{t|t-1}^{-1})_{ij} + (R_{t}^{-1})_{kl}\tilde{h}_{ij}^{kl}(x_{t}) 
\\
&- (R_{t}^{-1})_{kl}\Exp_{\mathcal{N}(x_t;\hat{x}_{t|t},P_{t|t})}\biggl\{\frac{\partial^{2} g^{k}}{\partial x^{i}\partial x^{j}}(x_t)\biggr\}\zeta_{t}^{l}
\\
= &(P_{t|t-1}^{-1})_{ij} + (R_{t}^{-1})_{kl}\biggl[\tilde{h}_{ij}^{kl}(\hat{x}_{t|t}) 
\\
&+ \frac{\partial\tilde{h}_{ij}^{kl}}{\partial x^{q}}(\hat{x}_{t|t})(x_{t}^{q} - \hat{x}_{t|t}^{q}) + \varphi_{ij}^{kl}(x_{t} - \hat{x}_{t|t})\biggr] 
\\
&- (R_{t}^{-1})_{kl} \Exp_{\mathcal{N}(x_t;\hat{x}_{t|t},P_{t|t})}\biggl\{\frac{\partial^{2} g^{k}}{\partial x^{i}\partial x^{j}}(x_t)\biggr\} \zeta_{t}^{l},
\end{aligned}
\end{equation}
where $$\varphi_{ij}^{kl}(x_{t}-\hat{x}_{t|t})\leq \kappa_{3} \|x_{t}-\hat{x}_{t|t} \|^{2}, \forall \; 0 \leq \|x_{t} - \hat{x}_{t|t}\| < \epsilon_{3},$$ for some $\kappa_{3},\epsilon_{3} > 0$, are high-order terms in the Taylor expansion of $\tilde{h}^{kl}_{ij}$.


The main idea of stability analysis is based on the application of Lyapunov functions, just as in the case of Kalman-type nonlinear filter \cite{reif1999stochastic,xu2016stochastic}. In order to construct the Lyapunov function, let us first introduce the auxiliary covariance matrix $\tilde{P}_{t|t}$ and $\tilde{P}_{t|t-1}$, which evolve according to
\begin{equation}\nonumber
\begin{aligned}
(\tilde{P}_{t|t}^{-1})_{ij} = &(\tilde{P}_{t|t-1}^{-1})_{ij} + (R_{t}^{-1})_{kl}\tilde{h}_{ij}^{kl}(\hat{x}_{t|t}),
\\
(\tilde{P}_{t|t-1})^{ij} = &\frac{\partial f^{i}}{\partial x^{k}}(\hat{x}_{t-1|t-1})\frac{\partial f^{j}}{\partial x^{l}}(\hat{x}_{t-1|t-1})
\\
&\times(\tilde{P}_{t-1|t-1})^{kl}+ (Q_{t})^{ij},
\end{aligned}
\end{equation}
and in matrix form
\begin{equation}\label{eq.matrix form of auxiliary matrix}
\begin{aligned}
\tilde{P}_{t|t}^{-1} &= \tilde{P}_{t|t-1}^{-1} + D_{t},
\\
\tilde{P}_{t|t-1} &= F_{t-1}\tilde{P}_{t-1|t-1}F_{t-1}^{\top} + Q_{t},
\end{aligned}
\end{equation}
where $D_{t}$ is the matrix with each component $(D_{t})_{ij} = (R_{t}^{-1})_{kl}\tilde{h}_{ij}^{kl}(\hat{x}_{t|t})$. 
Note that the evolution of $\tilde{P}_{t|t}$ and $\tilde{P}_{t|t-1}$ does not depend directly on the errors $e_{t|t}$, $e_{t|t-1}$ or the noise terms $\xi_{t}$, $\zeta_{t}$. This makes the following positive definiteness and boundedness assumption largely a condition on the system itself, much like the detectable and controllable conditions for linear systems \cite{jazwinski2007stochastic}.
\begin{assumption}\label{assump.boundness of P}
There exist constants $\underline{p},\bar{p} > 0$, such that
\begin{equation}\nonumber
\begin{aligned} 
&0<\underline{p}I\leq \tilde{P}_{t|t}\leq \bar{p}I < \infty,\ \forall \ t\geq 0,
\\
&0<\underline{p}I\leq \tilde{P}_{t|t-1}\leq \bar{p}I < \infty,\ \forall \ t\geq 0.
\end{aligned}
\end{equation}
\end{assumption}

Our main stability result is stated in the following theorem. Generally speaking, this theorem proves the stability of our proposed method for those systems with almost linear measurement functions and small noise.
\begin{theorem}
Under Assumption \ref{assump.bounded derivative} and \ref{assump.boundness of P}, the estimation error $e_{t|t}$ is exponentially bounded in the mean square for systems with almost linear measurement functions, i.e., there exist $\epsilon,\epsilon',\lambda > 0$, such that
\begin{equation}\nonumber
\Exp \left \|e_{t|t} \right\|^{2}\leq \epsilon \|e_{0|0}\|^{2}\biggl(\frac{1}{1+\lambda}\biggr)^{t} + \epsilon',\ \forall \ t\geq 0,
\end{equation} 
as long as the initial error and the strength of the system noise are small enough, that is $\|e_{0|0}\|\leq \delta$, $\Exp\left\{\xi_{t}\xi_{t}^{\top} \right\} \leq \delta I$ and $\Exp\left\{\zeta_{t}\zeta_{t}^{\top} \right\} \leq \delta I$ for some given $\delta > 0$. 
\end{theorem}
\begin{proof}
Utilizing \eqref{eq.evolution of e} and \eqref{eq.matrix form of auxiliary matrix}, $e_{t|t}^{\top} \tilde{P}_{t|t}^{-1}e_{t|t}$ can be computed as follows:
\begin{equation}
\begin{aligned} 
&e_{t|t}^{\top}\tilde{P}_{t|t}^{-1}e_{t|t} 
\\
=& e_{t-1|t-1}^{\top}F_{t-1}^{\top}(I+H_{t})^{-1}
\\
&\quad\quad\times\left[\left(F_{t-1}\tilde{P}_{t-1|t-1}F_{t-1}^{\top} +Q_{t}\right)^{-1} + D_{t}\right]
\\
&\quad\quad \times(I+H_{t})^{-1}F_{t-1}e_{t-1|t-1}
\\
&\quad - 2e_{t-1|t-1}^{\top}F_{t-1}^{\top}(I+H_{t})^{-1}\tilde{P}_{t|t}^{-1}(I+H_{t})^{-1}\bar{h}_{t}
\\
&\quad + \bar{h}_{t}^{\top}(I+H_{t})^{-1}\tilde{P}_{t|t}^{-1}(I+H_{t})^{-1}\bar{h}_{t} 
\\
&\quad+ \Psi_{t}^{\top}(I+H_{t})^{-1}\tilde{P}_{t|t}^{-1}(I+H_{t})^{-1}\Psi_{t}
\\
&\quad + \left(\xi_{t}-G_{t}\zeta_{t}\right)^{\top}(I+H_{t})^{-1}\tilde{P}_{t|t}^{-1}(I+H_{t})^{-1} \left(\xi_{t}-G_{t}\zeta_{t}\right)
\\
&\quad + 2e_{t-1|t-1}F_{t-1}^{\top}(I+H_{t})^{-1}
\\
&\qquad\times \tilde{P}_{t|t}^{-1}(I+H_{t})^{-1}\left(\Psi_{t} + \xi_{t} - G_{t}\zeta_{t}\right)
\\
&\quad - 2\bar{h}_{t}^{\top}(I+H_{t})^{-1}\tilde{P}_{t|t}^{-1}(I+H_{t})^{-1}\left(\Psi_{t} + \xi_{t} - G_{t}\zeta_{t}\right)
\\
&\quad + 2\Psi_{t}^{\top}(I+H_{t})^{-1}\tilde{P}_{t|t}^{-1}(I+H_{t})^{-1}\left(\xi_{t} - G_{t}\zeta_{t}\right).
\end{aligned}
\label{eq:51}
\end{equation}

For the first term on the right-hand side of \eqref{eq:51}, with the fact that the measurement function $h(x)$ is almost linear, there exists $\lambda_{1} > 0$ such that
\begin{equation}
\begin{aligned} 
F_{t-1}^{\top}&(I+H_{t})^{-1}\left[\left(F_{t-1}\tilde{P}_{t-1|t-1}F_{t-1}^{\top} +Q_{t}\right)^{-1} + D_{t}\right]
\\	&\times(I+H_{t})^{-1}F_{t-1}\leq \frac{1}{1+\lambda_{1}}\tilde{P}_{t-1|t-1}.
\end{aligned} 
\label{eq:52}
\end{equation}
In fact, if $g(x) = Cx$ is a linear function with some constant matrix $C\in\mathbb{R}^{m\times n}$, then
\begin{equation}\nonumber
\begin{aligned} 
& H_{t} = P_{t|t-1}C^{\top}R_{t}^{-1}C,
\\
& D_{t} = C^{\top}R_{t}^{-1}C,
\end{aligned} 
\end{equation}
and $\tilde{P}_{t|t} = P_{t|t}$, $\tilde{P}_{t|t-1} = P_{t|t-1}$. In this way, we have
\begin{equation}\nonumber
\begin{aligned} 
&F_{t-1}^{\top}(I+H_{t})^{-1}\left[\left(F_{t-1}\tilde{P}_{t-1|t-1}F_{t-1}^{\top} +Q_{t}\right)^{-1} + D_{t}\right]
\\
&\qquad\times(I+H_{t})^{-1}F_{t-1}
\\
=& F_{t-1}^{\top}(I+P_{t|t-1}C^{\top}R_{t}^{-1}C)^{-1}P_{t|t-1}^{-1}F_{t-1}
\\
=& F_{t-1}^{\top}(I+P_{t|t-1}C^{\top}R_{t}^{-1}C)^{-1}(F_{t-1}^{\top}P_{t-1|t-1}F_{t-1}^{\top})^{-1}
\\
&\quad\times(I + (F_{t-1}P_{t-1|t-1}F_{t-1}^{\top})^{-1}Q_{t})^{-1}F_{t-1}.
\end{aligned} 
\end{equation}
By observing that
\begin{align*} 
(I+P_{t|t-1}C^{\top}R_{t}^{-1}C)^{-1}&\leq \frac{1}{1+ \lambda}I,
\\
(I + (F_{t-1}P_{t-1|t-1}F_{t-1}^{\top})^{-1}Q_{t})^{-1} 
& \leq \frac{1}{1+\lambda'}I,
\end{align*}
with $\lambda = \lambda_{\min}(P_{t|t-1}C^{\top}R_{t}^{-1}C) > 0 $ and $\lambda' = \lambda_{\min}((F_{t-1}P_{t-1|t-1}F_{t-1}^{\top})^{-1}Q_{t}) > 0$, therefore \eqref{eq:52} holds for some $\lambda_{1} > 0$ for linear systems. Because of the continuity of \eqref{eq:52} with respect to the measurement function $g$, \eqref{eq:52} also holds for those systems with almost linear measurement functions.

Other terms on the right-hand side of \eqref{eq:51} can also be bounded. Firstly, there exists $\delta_{1} > 0$, such that
\begin{equation}
\bar{h}_{t}^{\top}(I+H_{t})^{-1}\tilde{P}_{t|t}^{-1}(I+H_{t})^{-1}\bar{h}_{t} \leq \delta_{1}.
\label{eq:55}
\end{equation}
Because $\bar{h}_{t} \equiv 0$ for linear measurement functions, the constant $\delta_{1}$ above can be chosen to be sufficiently small for almost linear measurement functions.

Secondly, there exist $\delta_{2} > 0$, such that
\begin{equation}
- 2e_{t-1|t-1}^{\top}F_{t-1}^{\top}(I+H_{t})^{-1} \tilde{P}_{t|t}^{-1}(I+H_{t})^{-1}\bar{h}_{t}\leq \delta_{2} \|e_{t-1|t-1}\|.
\end{equation}

Finally, because $\xi_{t}$, $\zeta_{t}$ are independent Gaussian random variables, there exists $\delta_{3} > 0$, such that
\begin{equation}
\Exp \left\{(\xi_{t}-G_{t}\zeta_{t})^{\top}(I+H_{t})^{-1} \tilde{P}_{t|t}^{-1}(I+H_{t})^{-1}(\xi_{t}-G_{t}\zeta_{t}) \right\}\leq \delta_{3}. 
	\label{eq:57}
\end{equation}

Other terms on the right-hand side of \eqref{eq:51} are related to the high order terms $\Psi_{t}$. Therefore, combining \eqref{eq:52}, \eqref{eq:55} to \eqref{eq:57}, there exists $\eta > 0$, such that if $\|e_{t|t}\| < \eta$ for all $t\geq 0$, then
\begin{equation}\nonumber
\begin{aligned} 
&\Exp\left\{\left.e_{t|t}^{\top}\tilde{P}_{t|t}^{-1}e_{t|t}\right|e_{t-1|t-1}\right\}
\\
\leq& \frac{1}{1+\lambda_{0}}e_{t-1|t-1}^{\top}\tilde{P}_{t-1|t-1}^{-1}e_{t-1|t-1} + \delta_{2} \cdot \Exp \|e_{t-1|t-1}\| + \delta,
\end{aligned} 
\end{equation}
for some $0 <\lambda_{0} <\lambda_{1}$.

With the fact that $\underline{p}I\leq \tilde{P}_{t|t}\leq \bar{p}I$, there exist a vector $\epsilon_{1}\in\mathbb{R}^{n}$ and a constant $   \epsilon_{2} > 0$, such that 
\begin{equation}
\begin{aligned} 
&\Exp \left\{\left.(e_{t|t}-\epsilon_{1})^{\top}\tilde{P}_{t|t}^{-1}(e_{t|t}-\epsilon_{1})\right|e_{t-1|t-1}\right\}
\\
\leq& \frac{1}{1+\lambda_{0}}(e_{t-1|t-1}-\epsilon_{1})^{\top}\tilde{P}_{t-1|t-1}^{-1}(e_{t-1|t-1}-\epsilon_{1}) + \epsilon_{2},
	\end{aligned} 
\label{eq:59}
\end{equation}
as long as $\|e_{t-1|t-1}\| < \eta$.

Also, if there exists $\tilde{\eta} > 0$, such that $\tilde{\eta}<\|e_{t-1|t-1}\|<\eta$, for $\epsilon_{1}$ with sufficiently small norm and sufficiently small $\epsilon_{2}$, we have the following supermartingale-like property:
\begin{equation}\nonumber
\begin{aligned}
&\Exp \left\{\left.(e_{t|t}-\epsilon_{1})^{\top}\tilde{P}_{t|t}^{-1}(e_{t|t}-\epsilon_{1})\right|e_{t-1|t-1}\right\}
\\
&- (e_{t-1|t-1}-\epsilon_{1})^{\top}\tilde{P}_{t-1|t-1}^{-1}(e_{t-1|t-1}-\epsilon_{1})
\\
\leq& -\frac{\lambda_{0}}{2(1+\lambda_{0})\bar{p}}\tilde{\eta}^{2} + \epsilon_{2} < 0.
	\end{aligned}
\end{equation}

Hence, if the initial estimation error $\|e_{0|0}\|<\eta$, then we can recursively use \eqref{eq:59} to compute the estimation error at time $t$,
\begin{equation}\nonumber
\begin{aligned} 
&\Exp \left\{(e_{t|t}-\epsilon_{1})^{\top}\tilde{P}_{t|t}^{-1}(e_{t|t}-\epsilon_{1}) \right\}
\\
\leq& \epsilon_{2}\sum_{k=0}^{t-1}\frac{1}{(1+\lambda_{0})^{k}} + \frac{1}{\underline{p}}\|e_{0|0} - \epsilon_{1}\|^{2}\left(\frac{1}{1+\lambda_{0}}\right)^{t}\\
\leq&  \frac{1+\lambda_{0}}{\lambda_{0}}\epsilon_{2} + \frac{1}{\underline{p}}\|e_{0|0} - \epsilon_{1}\|^{2}\left(\frac{1}{1+\lambda_{0}}\right)^{t}.
	\end{aligned} 
\end{equation}
Since
\begin{equation}\nonumber
\frac{1}{\bar{p}}\|e_{t|t}-\epsilon_{1}\|^{2}\leq (e_{t|t}-\epsilon_{1})^{\top}\tilde{P}_{t|t}^{-1}(e_{t|t}-\epsilon_{1}),
\end{equation}
and thus,
\begin{equation}\nonumber
\|e_{t|t}\|\leq \sqrt{\bar{p}(|e_{t|t}|-\epsilon_{1})^{\top}\tilde{P}_{t|t}^{-1}(|e_{t|t}|-\epsilon_{1})} + |\epsilon_{1}|,
\end{equation}
and we obtained the exponentially bounded in mean square for $e_{t|t}$, i.e., there exist $\epsilon, \epsilon' > 0$, such that
\begin{equation}\nonumber
\Exp \|e_{t|t}\|^{2}\leq \epsilon\|e_{0|0}\|^{2}\biggl(\frac{1}{1+\lambda_{0}}\biggr)^{t} + \epsilon',\ \forall \ t\geq 0.
\end{equation}
\end{proof}

\section{Robustifying NANO Filter with Gibbs Posterior}\label{sec.robust variants}

In the update step of Bayesian filtering, as illustrated in \eqref{eq.BF update optimizaton}, the objective is to strike a balance between the information provided by the measurement data and the prior distribution. However, obtaining reliable information from the measurement data requires a comprehensive probabilistic model of the measurement-generating process. In other words, a precise specification of the output distribution \( p(y_t|x_t) \) is needed beforehand. Unfortunately, in practical scenarios, it is often challenging to specify the true model of \( p(y_t|x_t) \) due to factors such as sensor malfunctions or unmodeled system dynamics. This leads to a mismatch between the true data-generating process and the assumed output probability model. Measurements influenced by such model misspecifications, which are typically referred to as outliers, require special attention to maintain the reliability of state estimates \cite{cao2023generalized,boustati2020generalised}. One effective approach to handle this issue is to replace the likelihood function in \eqref{eq.BF update optimizaton} with a generalized measurement-dependent loss function, leading to the following variational problem:
\begin{equation}\label{eq.Gibbs optimization}
\begin{aligned}
p_\mathrm{G}(x_t|y_{1:t}) =& \argmin_{q(x_t)} \Big\{\Exp_{q(x_t)} \left\{\ell^G(x_t, y_t) \right\}\\
&+ D_{\mathrm{KL}}\left\{q(x_t) || p(x_t|y_{1:t-1})\right\} \Big\}.    
\end{aligned}    
\end{equation}
Here, $\ell^G: \mathbb{R}^n \times \mathbb{R}^m \to \mathbb{R}$ is the generalized measurement-dependent loss function. It turns out that the solution to \eqref{eq.Gibbs optimization} is known as the Gibbs posterior, which has an analytical form akin to the standard Bayesian posterior:
\begin{equation}\nonumber
\begin{aligned}
p_\mathrm{G}(x_t|y_{1:t}) =  \frac{\exp\{-\ell^G(x_t, y_t)\} p(x_t|y_{1:t-1})}{\int \exp\{-\ell^G(x_t, y_t)\} p(x_t|y_{1:t-1}) \d x_t }.
\end{aligned}   
\end{equation}
Here, we denote the Gibbs posterior as $p_\mathrm{G}(x_t|y_{1:t})$ to differentiate from the standard Bayesian posterior $p(x_t|y_{1:t})$. In fact, the standard Bayesian posterior can be regarded as a special case of the Gibbs posterior by setting the loss function as the negation of the logarithm of the likelihood function, i.e., $\ell^G(x_t, y_t) = -\log p(y_t|x_t)$.

\begin{remark}
A notable fact is that the analysis of the Gaussian approximation in Section \ref{sec.stationary points of J} and the derivation of the NANO filter in Section \ref{sec.NANO} do not depend on the specific form of the loss function $\ell(x_t, y_t)$. Therefore, since Gibbs Bayesian filtering in \eqref{eq.Gibbs optimization} only modifies the loss function $\ell(x_t, y_t)$ in \eqref{eq.BF update optimizaton} to a more general form, $\ell^G(x_t, y_t)$, all of the previous analysis and algorithmic design can naturally extend to this case. 
\end{remark}

Then, we provide several potential choices of the measurement-dependent loss functions that can tackle model misspecifications:

\noindent 
\textbf{Choice 1: 
Composite Likelihood:}
One popular choice is to use the composite likelihood loss, i.e., combine multiple likelihood functions to achieve robust and adaptive performance under varying conditions. For example, for single-dimension measurement case, we can assume that the loss function is a composition of Gaussian likelihood and Laplace likelihood:
\begin{equation}\nonumber
\begin{aligned}
&\ell^{\mathrm{h}}(x_t, y_t) 
\\
=&\begin{cases} 
-\log{\mathcal{N}(y_t; g(x_t), 1)} & \text{if} \quad |y_t - g(x_t)| \leq \delta,
\\
-\log{\mathrm{Laplace}(y_t; g(x_t), \frac{1}{\delta})} & \text{otherwise.}
\end{cases}
\\
\stackrel{c}{=}&
\begin{cases}
\frac{1}{2} |y_t - g(x_t)|^2 & \text{if} \quad |y_t - g(x_t)| \leq \delta,
\\
\delta \left(|y_t - g(x_t)| - \frac{1}{2} \delta\right) & \text{otherwise.}
\end{cases}
\end{aligned}    
\end{equation}
Here, the notation $``\stackrel{c}{=}"$ represents that two expressions are equivalent up to an additive constant, and $\delta>0$ is the threshold variable. 
In fact, $\ell^{\mathrm{h}}(x_t, y_t)$ represents the famous Huber loss \cite{huber1992robust} used in robust regression, which is known to be less sensitive to outliers in data than the squared error loss. In practical applications, we utilize a smooth approximation of the Huber loss, known as the Pseudo-Huber loss, to improve optimization efficiency and ensure differentiability. The Pseudo-Huber loss is defined as:
\begin{equation}\label{eq.pseduo huber}
\begin{aligned}
&\ell^{\mathrm{ph}}(x_t, y_t) = \delta^2 \left( \sqrt{1 + \left(y_t-g(x_t)\right)^2/{\delta}^2} -1 \right).
\end{aligned}    
\end{equation}

\noindent \textbf{Choice 2: Weighted Log-Likelihood}:
Besides the composition of different likelihoods, one natural choice to achieve a robust loss function is to scale the log-likelihood loss function with a data-dependent weighting term:
\begin{equation}\label{eq.weighted likelihood}
\ell^\mathrm{w}(x_t, y_t) = -w(x_t, y_t) \cdot \log{p(y_t|x_t)}.    
\end{equation}
Here, $w: \mathbb{R}^n \times \mathbb{R}^m \to \mathbb{R}_{+}$ is the weighting function. The philosophy behind this form of loss is that reweighting the effect of outliers in the filtering procedure can potentionally improve robustness. The choices of the weighting function can be inspired from several domains. For example, we can use the so-called inverse multi-quadratic weighting function \cite{duran2024outlier} 
\begin{equation}\nonumber
w(x_t, y_t) = (1+ \left \|y_t-g(x_t)\right\|_{R_t^{-1}}^2/c^2 )^{-1},
\end{equation}
where $c>0$ is a constant number. Other choices of the weighting function can also be found in Section 3.3 of \cite{duran2024outlier}. 
\begin{remark}
When $w(x_t, y_t)$ is chosen as the constant value that is smaller than 1, it can be regarded as the result of performing the so-called exponential density
rescaling for convlutional Bayes filter \cite{cao2024convolutional}, which already shows robustness for systems with outliers. In this case, $w(x_t, y_t)$ can be regarded as the Lagrange multiplier which balances the compression and recontruction in the well-known information bottleneck problem \cite{cao2024convolutional}. 
\end{remark}

\noindent \textbf{Choice 3: 
Divergence-Dependent Loss Function:} 
The final option is to utilize robust divergence to fit the data. This approach is based on the fact that minimizing the negative log-likelihood essentially amounts to minimizing the KL divergence between the true data distribution $p_{\mathrm{true}}(y_t)$, which includes outliers, and the assumed likelihood:
\begin{equation}\label{eq.cross entropy}
\begin{aligned}
&\argmin_{q(x_t)} \Exp_{q(x_t)} \left\{D_{\mathrm{KL}}\left(p_{\mathrm{true}}(y_t) \| p(y_t|x_t) \right) \right\} \\
=& \argmin_{q(x_t)} \Exp_{q(x_t)} \Exp_{p_{\mathrm{true}}(y_t)} \left\{-\log{p(y_t|x_t)}\right\}
\\
\approx & \argmin_{q(x_t)} \Exp_{q(x_t)} \left\{-\log{p(y_t|x_t)}\right\}.
\end{aligned}
\end{equation}
In this equation, the first equality arises because the term $p_{\mathrm{true}}(y_t)$ is treated as a constant regarding $q(x_t)$ and thus omitted, since the objective is to find the minimizer, not the minimum value. The final approximation holds because we cannot directly access the true data distribution; instead, we use the sample $y_t$ to approximate the expected value. 

Based on this understanding, one natural idea for designing a robust loss is to replace the KL divergence with a more robust divergence, such as 
$\beta$-divergence or 
$\gamma$-divergence. For instance, replacing the KL divergence with 
$\beta$-divergence in \eqref{eq.cross entropy} leads to the $\beta$ loss function $\ell^\beta(x_t,y_t)$ \cite{cao2023generalized}:
\begin{equation}\label{eq.beta likelihood}
\ell^\beta(x_t,y_t)=-\frac{\beta+1}{\beta}
p(y_t|x_t) ^{\beta}
+\int{p(y|x_t)^{\beta+1}\mathrm{d}y}.   
\end{equation}

\begin{remark}
Formally analyzing which loss function provides superior robustness is challenging as each is grounded in different principles. A common approach is to apply Huber's robust statistics theory \cite{huber2011robust}, which assesses the robustness of an estimator using its influence function. An estimator is considered robust if its influence function remains bounded as an outlier's value increases indefinitely. Previous studies have conducted case-specific analyses of robustness for Huber loss \cite{gandhi2009robust}, weighted loss \cite{duran2024outlier}, and $\beta$ loss \cite{cao2023generalized} in linear Gaussian systems, demonstrating that adjusting the loss function can improve robustness. Since this paper focuses on applying robust loss functions to enhance the NANO filter, a detailed formal analysis is left for future work.
\end{remark}

\section{Discussions}\label{sec.discussion}

\noindent \textbf{Natural Gradient Descent for Gaussian Approximation: } 
The use of natural gradient descent for finding optimal Gaussian approximations is well-documented in the literature, with early works dating back to \cite{opper2009variational} with applications for Gaussian process regression. After that, related methods have been explored in various domains, including robot batch estimation \cite{barfoot2020exactly,barfoot2024state}, Bayesian deep learning \cite{khan2018fast,khan2018fast1}, approximate inference \cite{khan2017conjugate}, optimization \cite{khan2017variational} and more.

We utilize the natural gradient method to find the optimal Gaussian approximation for Bayesian filtering in this paper. The advantage of incorporating the natural gradient in Gaussian filtering may arise from the profound geometric properties of the statistical manifolds generated by the family of Gaussian distributions, as discussed in \cite{amari2012differential,lenglet2006statistics}. Intuitively, the computation of the natural gradient involves second-order derivatives of the probability density functions. Thus, geometric concepts such as Riemannian metrics, curvatures, and geodesics on the statistical manifold can be engaged and reflected in the algorithm, leading to better convergence performance. Nevertheless, a precise description of the relationship between the natural gradient and Gaussian distribution remains an important research direction for further exploration.

\noindent \textbf{Gradient-Based Gaussian Filters:} 
Gradient-based Gaussian filters, which use gradient descent and its variants to solve Gaussian filtering, were initially proposed in \cite{gultekin2017nonlinear}, utilizing Monte Carlo techniques to approximate the exact gradient. Subsequently, cubature rules were introduced to replace the Monte Carlo method for gradient approximation, and various optimization approaches, such as the conditional gradient method, the alternating direction method of multipliers, and the natural gradient descent were employed to replace basic gradient descent for optimization \cite{hu2022iterative, guo2022gaussian, guo2023recursive,yumei2022variational}. 
These methods have certain assumptions and limitations compared to our approach. For instance, (i) the methods in \cite{gultekin2017nonlinear, hu2022iterative, guo2022gaussian,yumei2022variational} are restricted to systems with Gaussian noise; (ii) \cite{gultekin2017nonlinear, guo2023recursive} and \cite{yumei2022variational} require performing linear approximations of the measurement model; (iii) \cite{hu2022iterative} modifies the implicit hard constraints of the original problem into soft constraints, rendering it inequivalent to the original problem and (iv) \cite{li2017bayesian} only finds the maximum a posterior estimate, rather than the entire optimal Gaussian distribution. Also, they both lack rigorous proofs of algorithm convergence and stability guarantees for the estimation error, which we provide in Section \ref{sec.theoretical analysis}. Besides, their optimization problems are constructed based on direct minimization of forward KL divergence between the candidate distribution and the true posterior, which does not support the extension to the Gibbs posterior as proposed in our method.

Moreover, we would like to emphasize that, while these works utilize gradient-based methods for Gaussian approximation, they lack a theoretical analysis of the stationary points related to the optimal Gaussian approximation. Specifically, they neither justify the use of gradient methods in the update process nor address the Gaussian approximation in the prediction step. In contrast, our paper provides a comprehensive analysis covering both the prediction and update steps, including a detailed examination of the stationary points for the optimal Gaussian approximation.

\noindent\textbf{Kalman Filtering as Natural Gradient Descent}: Corollary~\ref{corollary.convergence} shows that the canonical KF can be interpreted as a single iteration of the natural gradient in our proposed NANO filter. This result looks similar to the findings in \cite{YannOllivier}, where the equivalence between KF and online natural gradient descent is established. However, \cite{YannOllivier} focuses on parameter estimation, where the state is treated as a deterministic variable representing the parameters of the measurement model. In that context, the natural gradient is applied to maximize the likelihood with respect to the measurement model. In contrast, our result is built on a more general Bayeisan view, where the state is a random variable, and the natural gradient is used to find the optimal parameters of state distribution. Therefore, our result can be treated as a generalization of the result in \cite{YannOllivier}.

\section{Experiments}\label{sec.experiment}
In this section, we validate the proposed NANO filter and its robust variants using the loss functions introduced in Section \ref{sec.robust variants}. We conduct simulations for both linear and nonlinear systems, along with a real-world experiment
. All the simulations and experiment are evaluated by root mean square error (RMSE) which is defined by
\begin{equation}\nonumber
\mathrm{RMSE} = \sqrt{\frac{\sum_{t=1}^{T}\Vert x_t - \hat{x}_t \Vert^2}{n \cdot T}},  
\end{equation}
where $T$ is the total time step for each trajectory. As it is generally impossible to compare with all existing Gaussian filters, we focus on the most popular ones as baselines for our analysis, including the KF, EKF, UKF, IEKF, and PLF. These filters are all mentioned  repeatedly in the well-regarded textbook \cite{sarkka2023bayesian}.

\subsection{Linear System Simulation: Wiener Velocity Model}
First, we perform simulations for Wiener velocity model, which is a canonical linear Gaussian model commonly employed for target tracking \cite{sarkka2023bayesian}. In this model, the state represents the position and velocity of a moving object in two dimensions. The state vector is defined as $x = \left[p_x \ p_y \ v_x \ v_y \right]^{\top}$, where $p_x$ and $p_y$ are the object's positions along the longitude and lateral directions, and $v_x$ and $v_y$ are the corresponding velocities. The measurements are direct, noisy observations of the position components. The state transition model with a time step $\Delta t = 0.1$ and the measurement model can be described as  
\begin{equation}
\nonumber
\begin{aligned}
x_{t+1} &= \begin{bmatrix}
1 & 0 & \Delta t & 0 \\
0 & 1 & 0 & \Delta t \\
0 & 0 & 1 & 0 \\
0 & 0 & 0 & 1
\end{bmatrix} x_t+\xi_t,
\\
y_t &= \begin{bmatrix}
1 & 0 & 0 & 0 \\
0 & 1 & 0 & 0
\end{bmatrix} x_t+\zeta_t .
\end{aligned}
\end{equation}
Here, $\xi_t \sim \mathcal{N}(\xi_t; 0, Q)$ and $\zeta_t \sim \mathcal{N}(\zeta_t; 0, R)$ are the process and measurement noises with covariance matrices given by
\begin{equation}\nonumber
Q=\begin{bmatrix}
\frac{\Delta t^3}{3} & 0 & \frac{\Delta t^2}{2} & 0 \\
0 & \frac{\Delta t^3}{3} & 0 & \frac{\Delta t^2}{2} \\
\frac{\Delta t^2}{2} & 0 & \Delta t & 0 \\
0 & \frac{\Delta t^2}{2} & 0 & \Delta t
\end{bmatrix}, R=\mathbb{I}_{2 \times 2} .
\end{equation}

\begin{figure}[!t]
\centering
\includegraphics[width=0.45\textwidth]{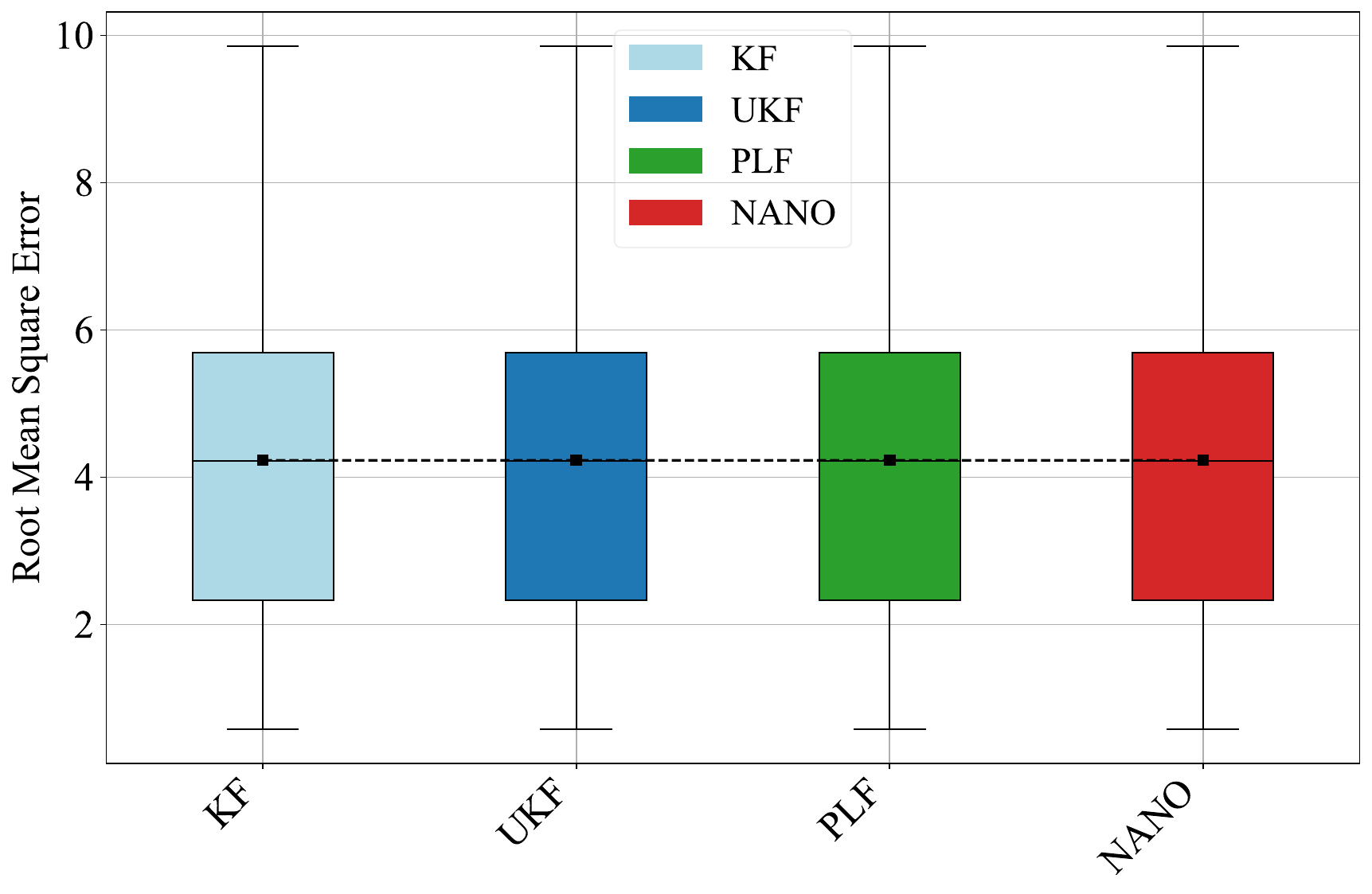}
\caption{Box plot of RMSE of KF, UKF, PLF and NANO filter, for the  standard Wiener velocity model. Note that the small black square `` \scalebox{0.4}{$\blacksquare$}
 " represents the average RMSE over all the Monte Carlo experiments.}
\label{fig.wiener}
\end{figure}

\begin{figure}[!t]
\centering
\includegraphics[width=0.45\textwidth]{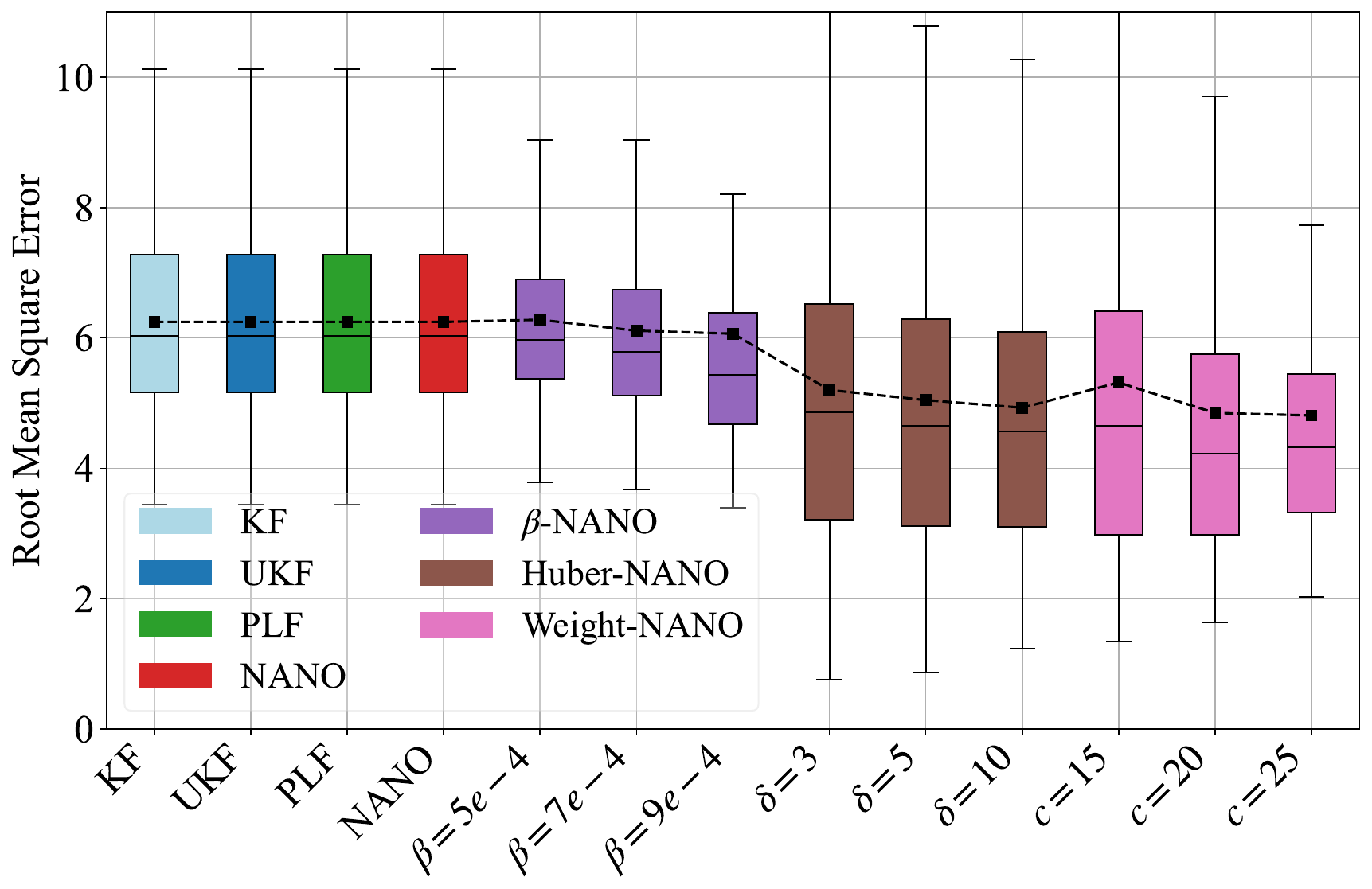}
\caption{Box plot of RMSE for KF, UKF, PLF, NANO filter, and its robust variants with various parameter values, for the Wiener velocity model with measurement outliers.}
\label{fig.wiener robust}
\end{figure}

To validate our method, we consider two scenarios. In the first scenario, the proposed NANO filter is compared against the KF, the optimal filter for linear Gaussian systems, as well as UKF and PLF. As shown in Fig.~\ref{fig.wiener}, the NANO filter achieves the same performance as the KF family within one iteration,  supporting the result in Corollary~\ref{corollary.convergence} that KF equals a single iteration of NANO filter.
In the second scenario, we consider the case where the measurement data deviates from the measurement model due to contamination of outliers. Specifically, the measurement data has a 10\% probability of being contaminated, modeled as $\zeta_t \sim 0.9 \cdot \mathcal{N}(0, R) + 0.1 \cdot \mathcal{N}(0, 1000R)$.
For this case, as illustrated in Fig.~\ref{fig.wiener robust}, we evaluate the robust variants of the NANO filter, including the Huber-NANO filter, Weight-NANO filter, and $\beta$-NANO filter, which utilize the robust loss functions $\ell^{\mathrm{ph}}$ \eqref{eq.pseduo huber}, $\ell^{\mathrm{w}}$ \eqref{eq.weighted likelihood}, and $\ell^{\beta}$ \eqref{eq.beta likelihood}, respectively.  The box plots indicate that the robust variants of the NANO filter generally outperform the standard NANO filter using different parameters. Notably, the Weight-NANO filter with $c = 25$ shows the best performance.

\subsection{Nonlinear System Simulation: Air-traffic Control}
Next, we perform simulations for a nonlinear system, focusing on an air-traffic control scenario where an aircraft executes a maneuvering turn in a horizontal plane at a height $h = 50$ with respect to the radar \cite{garcia2015posterior}. The state vector is defined as $x = \left[p_x \ v_x \ p_y \ v_y \ \omega \right]^{\top}$, where $p_x, p_y$ are the positions, $v_x, v_y$ are the velocities in the planar coordinates, and $\omega$ is the turn rate. The state transition model representing the kinematics of the turning motion and the measurement model are given by:
\begin{equation}
\nonumber
\begin{aligned}
x_{t+1} &= \begin{bmatrix}
1 & \frac{\sin \omega_t \Delta t}{\omega_t} & 0 & -\frac{1-\cos \omega_t \Delta t}{\Omega_t} & 0 \\
0 & \cos \omega_t \Delta t & 0 & -\sin \omega_t \Delta t & 0 \\
0 & \frac{1-\cos \omega_t \Delta t}{\omega_t} & 1 & \frac{\sin \omega_t \Delta t}{\omega_t} & 0 \\
0 & \sin \omega_t \Delta t & 0 & \cos \omega_t \Delta t & 0 \\
0 & 0 & 0 & 0 & 1
\end{bmatrix} x_t+\xi_t, \\
y_t &= \begin{bmatrix}
\sqrt{p_{x,t}^2+p_{y,t}^2+h^2} \\
\operatorname{atan}\left(\frac{p_{y,t}}{p_{x,t}}\right) \\
\operatorname{atan}\left(\frac{h}{\sqrt{p_{x,t}^2+p_{y,t}^2}}\right) \\
\frac{p_{x,t} v_{x,t}+p_{y,t} v_{y,t}}{\sqrt{p_{x,t}^2+p_{y,t}^2+h^2}}
\end{bmatrix}+\zeta_t . \\
\end{aligned}
\end{equation}

Here, $\Delta t = 0.2$ is the sampling period. The process noise $\xi_t \sim \mathcal{N}(0, Q)$ and measurement noise $\zeta_t \sim \mathcal{N}(0, R)$ have covariances given by:
\begin{equation}\nonumber
\begin{aligned}
Q &= \begin{bmatrix}
q_1 \cdot \frac{\Delta t^3}{3} & q_1 \cdot \frac{\Delta t^2}{2} & 0 & 0 & 0 & 
\\
q_1 \cdot \frac{\Delta t^2}{2} & q_1 \cdot \Delta t & 0 & 0 & 0
\\
0& 0& q_1 \cdot \frac{\Delta t^3}{3} & q_1 \cdot \frac{\Delta t^2}{2} & 0
\\
0& 0& q_1 \cdot \frac{\Delta t^2}{2} & q_1 \cdot \Delta t & 0
\\
0 & 0 & 0 & 0 & q_2 \cdot \Delta t
\end{bmatrix}, 
\\
R &= \mathrm{diag} \left( \left[ 
\sigma_r^2, \ \sigma_\varphi^2, \ \sigma_\theta^2, \ \sigma_{\dot{r}}^2 
\right] \right) .
\end{aligned}
\end{equation}

\begin{figure}[!t]
\centering
\includegraphics[width=0.45\textwidth]{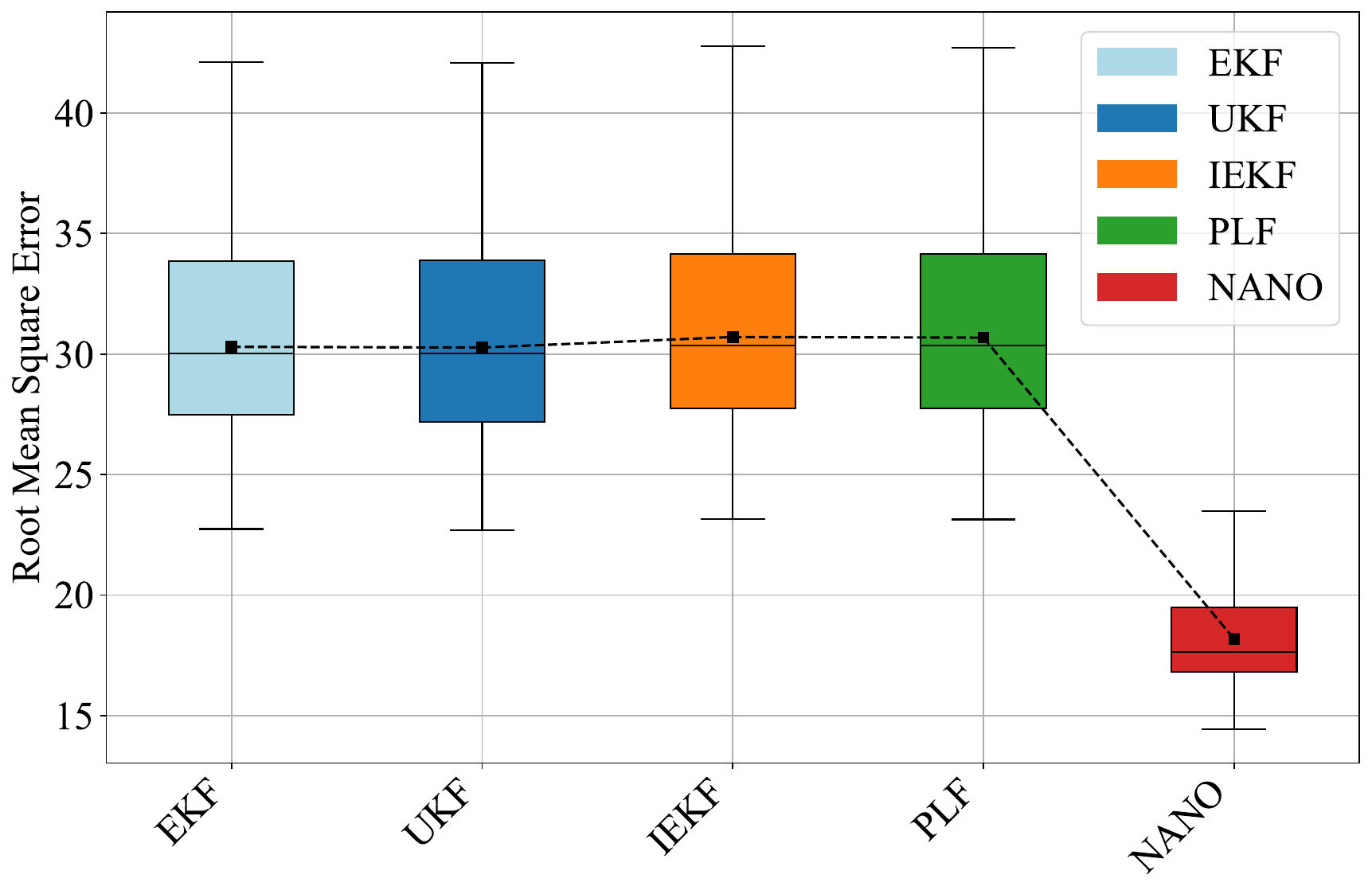}
\caption{Box plot of RMSE for EKF, UKF, IEKF, PLF and NANO filter, for the standard air-traffic control model.}
\label{fig.air traffic}
\end{figure}

\begin{figure}[!t]
\centering
\includegraphics[width=0.48\textwidth]{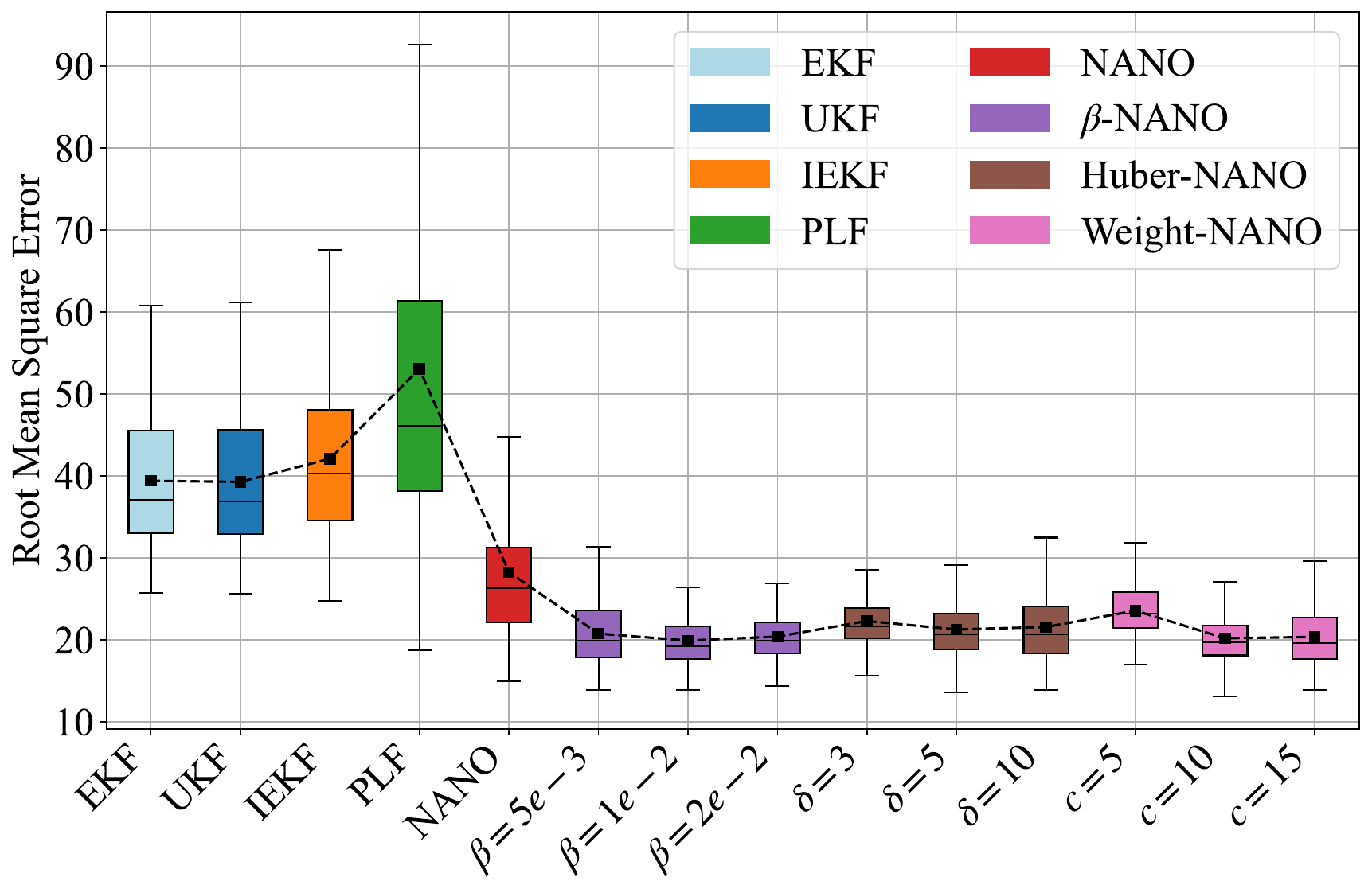}
\caption{Box plot of RMSE for EKF, UKF, IEKF, PLF, NANO filter, and its robust variants with various parameter values, for the air-traffic model with measurement outliers.}
\label{fig.air traffic outlier}
\end{figure}
Here, all noise parameters are the same as those in \cite{garcia2015posterior}. Similar to the linear case, we validate our method under two scenarios for the nonlinear system. In the first scenario, shown in Fig.~\ref{fig.air traffic}, the NANO filter is compared with  EKF, UKF, IEKF, and PLF. The results demonstrate that NANO filter significantly outperforms these established filters, achieving nearly a half reduction in RMSE. This substantial improvement highlights the NANO filter’s superior effectiveness in handling nonlinear dynamics.
In the second scenario, we introduce measurement outliers to reflect conditions where the measurement data deviates from the assumed model. Specifically, we assume the actual measurement noise follows a mixture model: $\zeta_t \sim 0.9 \cdot \mathcal{N}(\zeta_t; 0, R) + 0.1 \cdot \mathcal{N}(\zeta_t; 0, 100\cdot R)$, indicating a 10\% probability that measurements are contaminated by noise with a tenfold increase in standard deviation. As shown in Fig.~\ref{fig.air traffic outlier}, robust variants of the NANO filter—including the 
$\beta$-NANO, Huber-NANO, and Weight-NANO filters—significantly improve estimation accuracy in the presence of measurement outliers. Notably, the $\beta$-NANO filter with $\beta = 10^{-2}$ exhibits the best performance, effectively mitigating the impact of outliers.

\subsection{Real-world Experiment: Unmanned Ground Vehicle Localization}
We conducted a real-world experiment to demonstrate the effectiveness of our method using an unmanned ground vehicle (UGV) equipped with a Lidar sensor for environmental perception, as shown in Fig. \ref{fig.exp platform}. The vehicle's state is represented by $x=\left[\begin{matrix}p_x&p_y&\theta\end{matrix}\right]^\top$, which includes its 2D position ($p_x$ and $p_y$) and orientation angle ($\theta$). The longitudinal velocity $v$ and yaw rate $\omega$ serve as control inputs \cite{elfring2021particle}. The Lidar has a 240-degree detection range with 0.33-degree resolution. The UGV was manually controlled throughout the experiment.

\begin{figure}[!h]
\centering 
\subfloat{\includegraphics[width=0.4095 \linewidth]{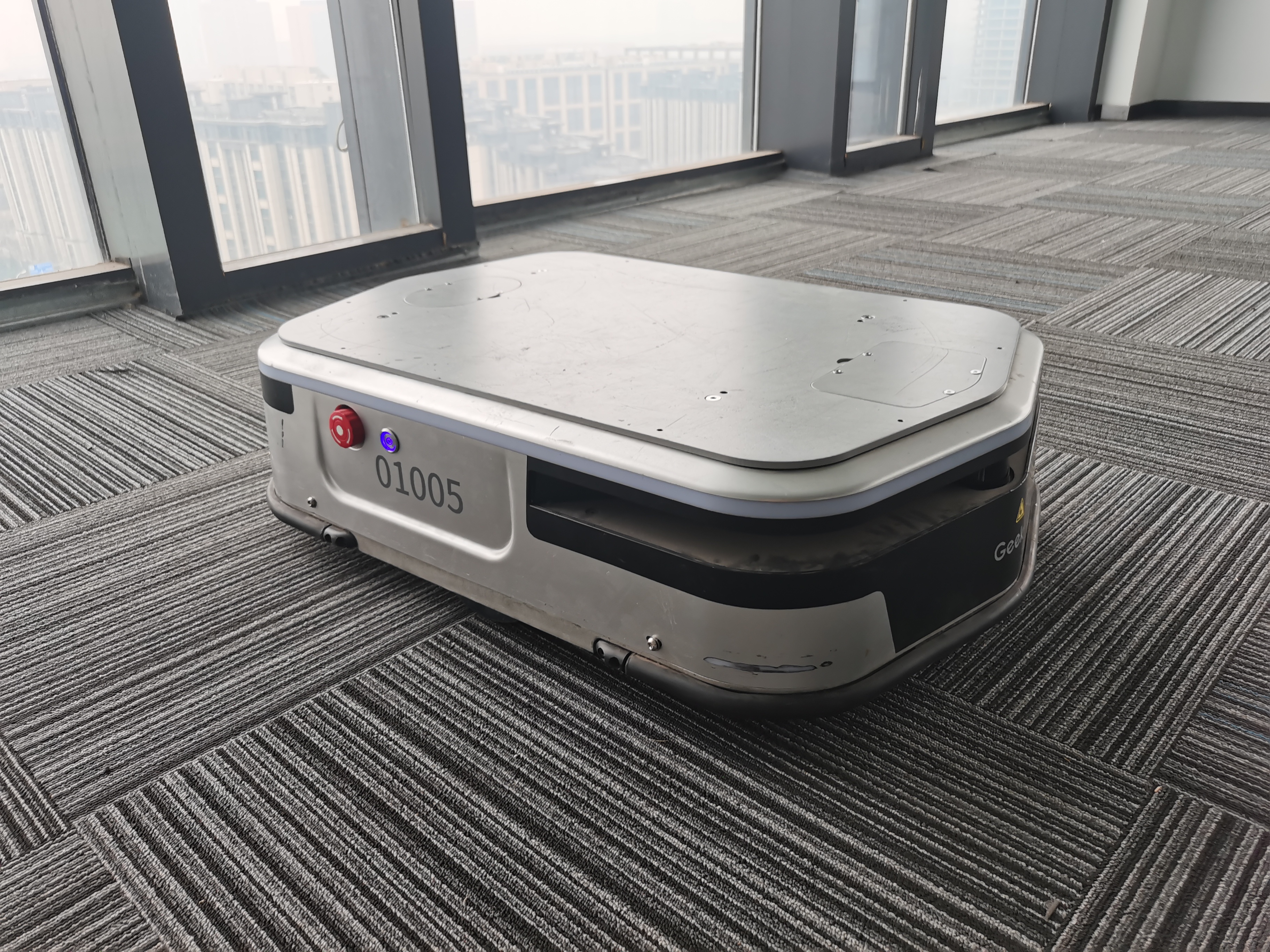}}\hspace{1mm} 
\subfloat{\includegraphics[width=0.441 \linewidth]{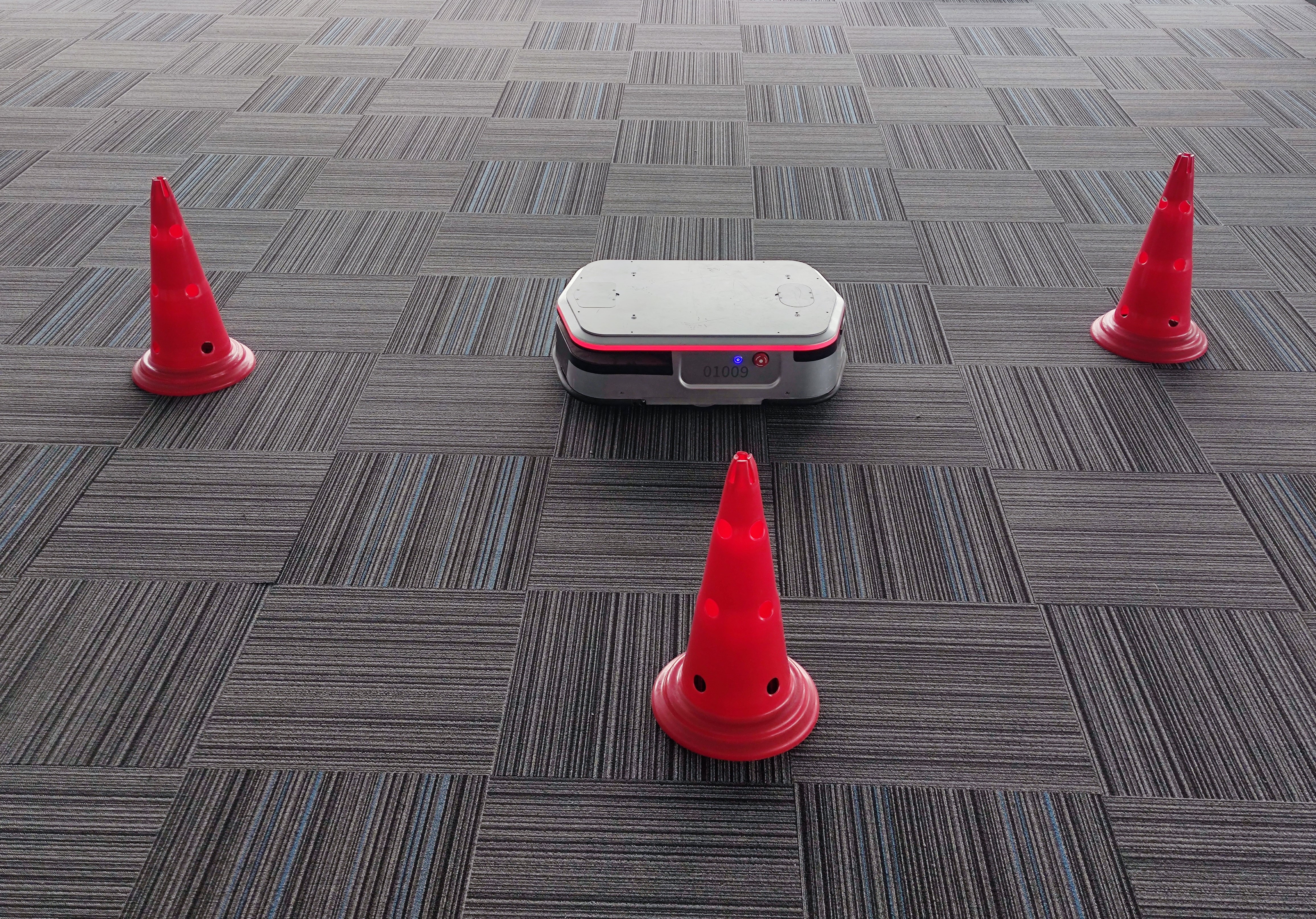}}
\caption{The UGV and the experiment field. The three red traffic cones serve as landmarks for positioning.}
\label{fig.exp platform}
\end{figure}

The vehicle kinematics are modeled as:
\begin{equation}
\nonumber
\begin{aligned}
    \begin{bmatrix}
        p_{x,t+1}\\
        p_{y,t+1}\\
        \theta_{t+1}
\end{bmatrix}=\begin{bmatrix}
        p_{x,t}\\
        p_{y,t}\\
        \theta_t
    \end{bmatrix} + \begin{bmatrix}
        v_t \cdot \cos\theta_t\\
        v_t \cdot \sin\theta_t\\
        \omega_t
    \end{bmatrix}\cdot\Delta t + \xi_t,
\end{aligned}
\end{equation}
where $\Delta t=0.0667$ is the sample period, $v_t$ is the longitudinal velocity, $\omega_t$ is the yaw rate, and $\xi_t$ is the process noise. The measurement model is given by:
\begin{equation}\nonumber
y_t = [d^1_t \ d^2_t \ d^3_t \ \alpha^1_t \ \alpha^2_t \ \alpha^3_t]^\top + \zeta_t,
\end{equation}
where $d^i$ and $\alpha^i$ ($i=1, 2, 3$) denote the relative distance and orientation angle between the UGV and each traffic cone:
\begin{equation}
\begin{aligned}
    d^i &= \|(p^{\mathrm{tc}, i}_{x}-p_x-l\cos\theta, p^{\mathrm{tc}, i}_{y}-p_y-l \sin\theta)\|_2,\\
    \alpha^i &= \arctan\left(\frac{p^{\mathrm{tc}, i}_{y}-p_y-l \sin\theta}{p^{\mathrm{tc}, i}_{x}-p_x-l\cos\theta}\right) - \theta.
\end{aligned}
\end{equation}
Here $(p^{\mathrm{tc}, i}_{x}, p^{\mathrm{tc}, i}_{y})$ is the position of the $i$-th traffic cone, and $l$ represents the longitudinal installation offset of the Lidar with respect to the robot center, as depicted in Fig. \ref{fig.robot system diagram}.

\hspace{-1cm}
\begin{figure}[!htb]
\centering 
\subfloat[]{\includegraphics[width=0.32\linewidth]{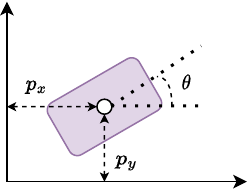}}\hspace{1mm} 
\subfloat[]{\includegraphics[width=0.38\linewidth]{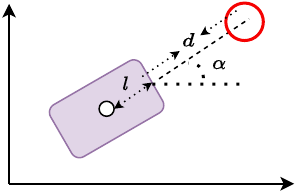}}
\caption{(a) The diagram of the vehicle's states containing the 2D position of the vehicle $p_x$ and $p_y$, and the orientation angle $\theta$. (b) The diagram of the measurement model. Note that the red circle represents the red traffic cone.}
\label{fig.robot system diagram}
\end{figure}
\begin{figure}[!t]
\centering
\includegraphics[width=0.48\textwidth]{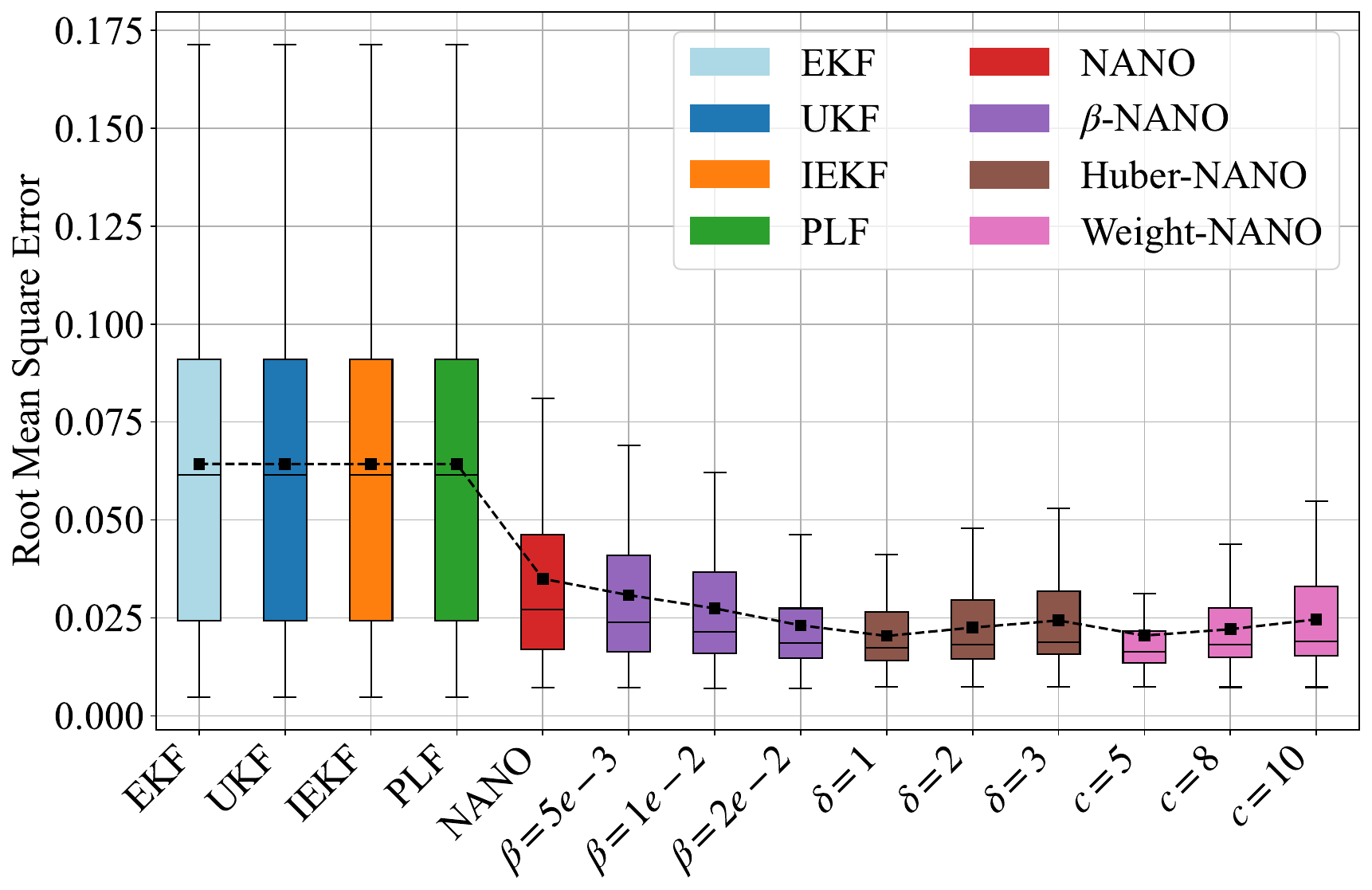}
\caption{Box plot of RMSE for KF, UKF, PLF, NANO and $\beta$-NANO, Huber-NANO, Weight-NANO with different values of the corresponding parameters.}
\label{fig.ugv}
\end{figure}

\begin{figure}[!t]
    \centering
    \subfloat{
        \includegraphics[width=0.22\textwidth]{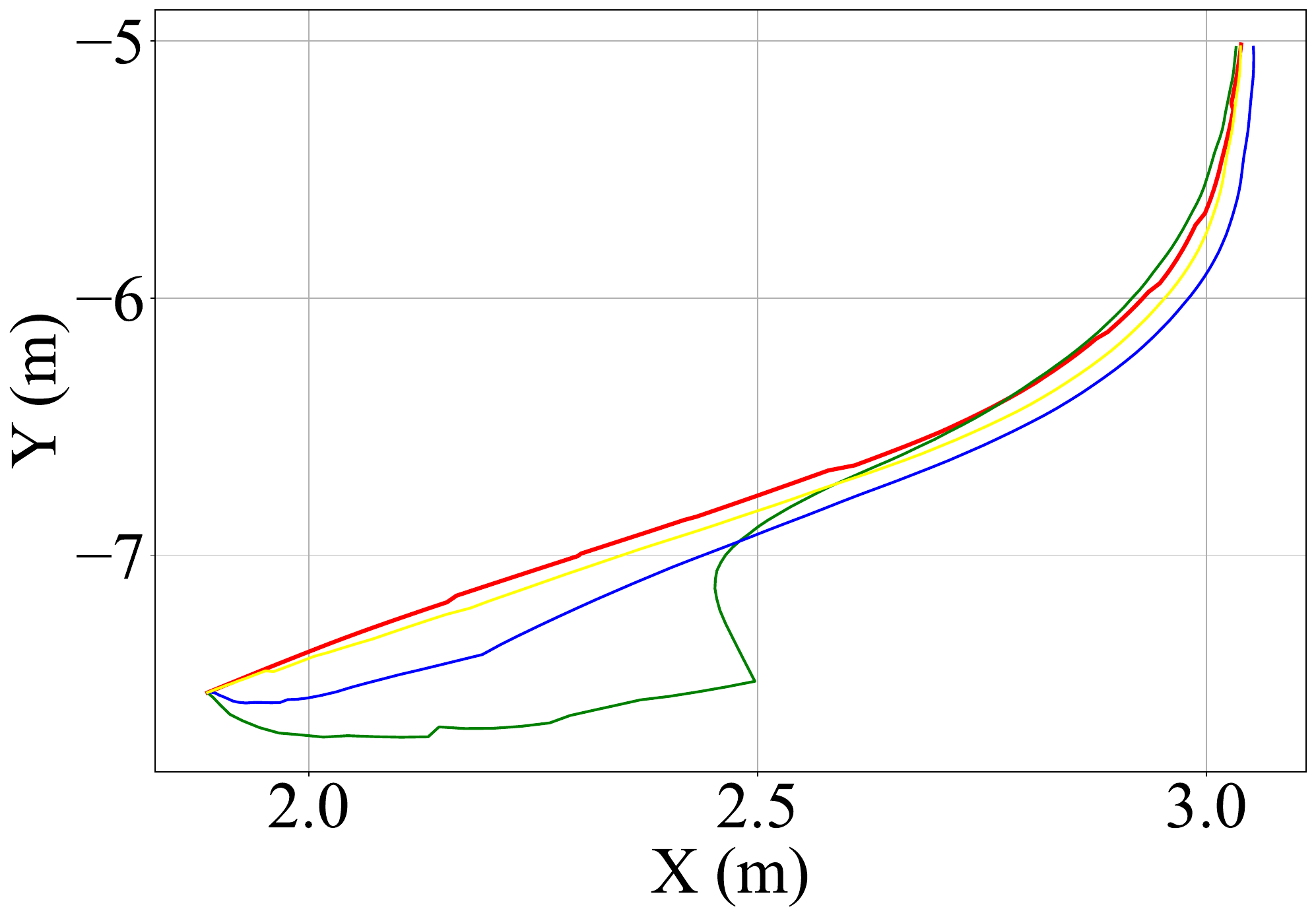}
        \label{ugv_time:fig1}
    }
    \hfill
    \subfloat{
        \includegraphics[width=0.23\textwidth]{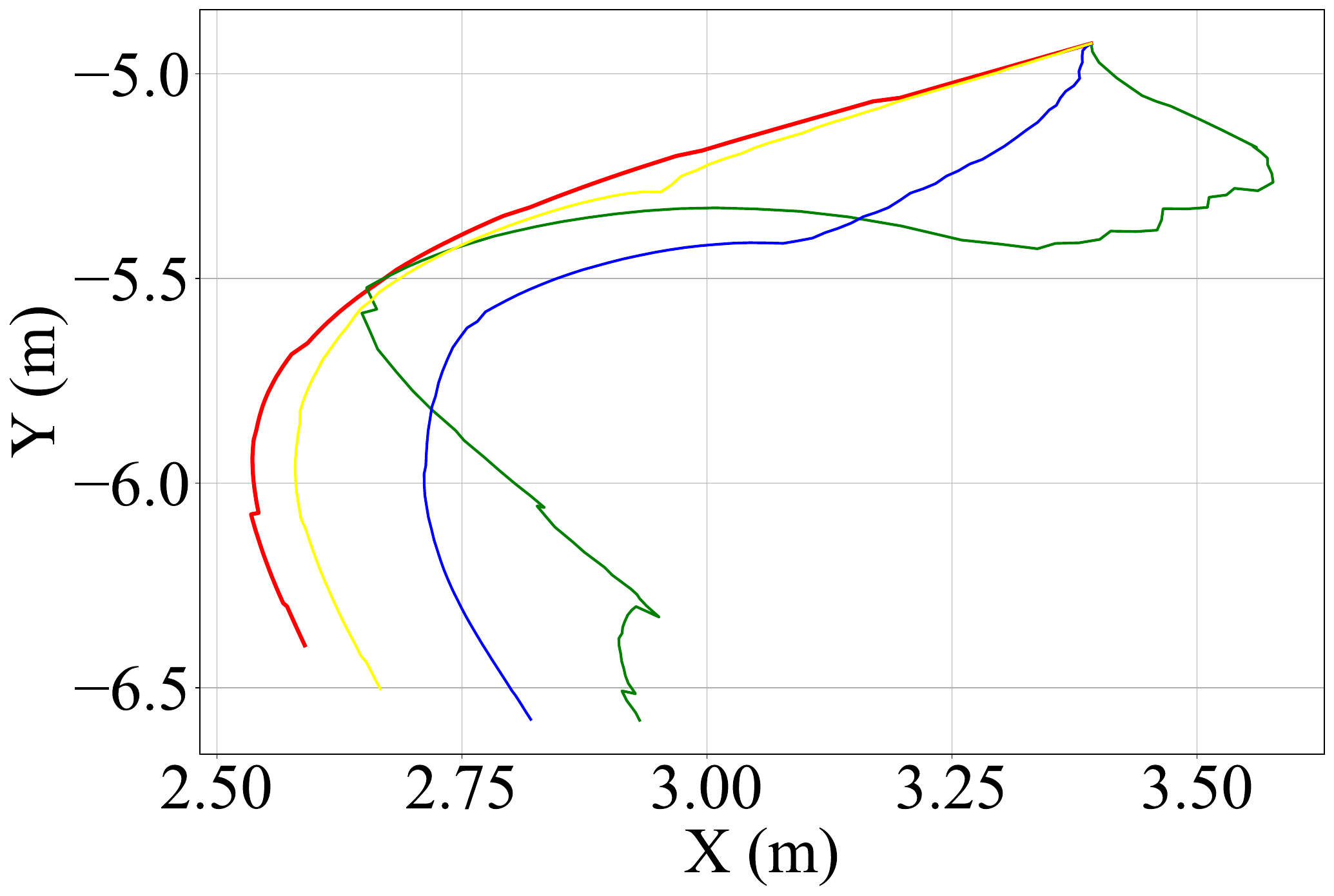}
        \label{ugv_time:fig2}
    }
    \\
    \subfloat{
        \includegraphics[width=0.45\textwidth]{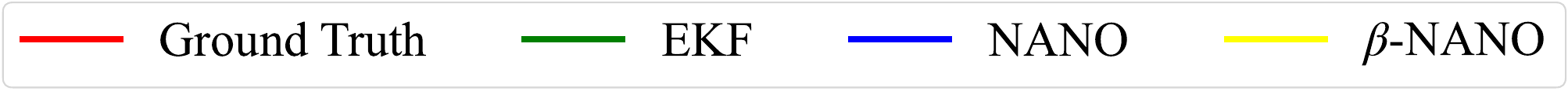}
        \label{ugv_time:legend}
        \phantomsection
    }
    \caption{Two samples of the real trajectories and their estimates.}
    \label{fig.ugv trajectory}
\end{figure} 

The experiment begin by collecting raw data, including ground truth states from the high-precision motion capture system, control inputs, and Lidar point clouds. The Lidar data are processed to extract relative distances and angles to the traffic cones, resulting in a dataset comprising trajectories totaling 11 minutes. The noises are assumed to be Gaussian, whose parameters are estimated from the ground truth and measurement data using maximum likelihood estimation. The dataset is then segmented into 50 sub-trajectories, each lasting 6.67 seconds (100 time steps). Figures \ref{fig.ugv} and \ref{fig.ugv trajectory} compare the RMSE performance of the EKF, UKF, IEKF, PLF, and NANO filters, including their robust variants. The results indicate that the NANO filter consistently outperforms traditional Gaussian filters, achieving higher accuracy. Furthermore, the robust NANO variants demonstrate improved performance, demonstrating their effectiveness in handling measurement outliers typically encountered in real-world conditions.

\begin{table*}[!t]
    \centering
    \caption{Run Time per time step of different algorithms on our experiments (ms)}
    \begin{tabular}{cccccc}
    \toprule
     Algorithms & Wiener Velocity &
     Wiener Velocity
     (outlier) &
     Air-traffic Control &
     Air-traffic Control
     (outlier) &
     UGV Localization\\
         \midrule
     EKF & 
     0.023 & 
     0.028 &
     1.456 &
     1.439 &
     0.809 \\
     UKF & 
     0.105 & 
     0.113 &
     0.712 &
     0.659 &
     0.447 \\
     iEKF & 
     0.023 & 
     0.039 &
     2.028 &
     2.015 &
     2.337 \\
     PLF & 
     0.264 & 
     0.275 &
     1.023 &
     1.313 &
     0.870 \\
     NANO filter & 
     0.406 & 
     0.441 &
     1.980 &
     2.023 &
     2.458 \\
    \bottomrule
    \end{tabular}

    \label{tab.time_run}
\end{table*}

Previous results for simulations and experiment demonstrate that the NANO filter significantly improves RMSE performance compared to traditional Gaussian filters. However, a critical aspect to consider is the computational cost per time step for each algorithm. Table \ref{tab.time_run} shows that for nonlinear systems, the average runtime of NANO filter is approximately 2 to 3 times that of the EKF, which indicates that despite being an iterative method, NANO filter maintains high efficiency and quickly converges. This balance of improved accuracy with reasonable computational demands highlights the practicality of the NANO filter, even in real-time applications.

\section{Conclusion}\label{sec.conclusion}
In this paper, we address the estimation errors commonly introduced by linearization techniques in Gaussian filters like the EKF and UKF for nonlinear systems. We reformulate the prediction and update steps of Gaussian filtering as optimization problems. While the prediction step mirrors moment-matching filters by calculating the first two moments of the prior distribution, the update step poses more challenges due to its highly nonlinear nature.
To overcome these issues, we propose an iterative approach called the NANO filter, which avoids linearization by directly optimizing the update step using the natural gradient derived from the Fisher information matrix. This allows us to account for the curvature of the parameter space and ensures more accurate updates. We prove that the NANO filter converges locally to the optimal Gaussian approximation at each time step and show that it provides exponential error bounds for nearly linear measurement models with low noise.
Experimental results demonstrate that the NANO filter outperforms popular filters like EKF, UKF, and others, while maintaining comparable computational efficiency.

\bibliographystyle{IEEEtran}
\bibliography{ref}

\vspace{-0.5cm}

\begin{IEEEbiography}[{\includegraphics[width=1in,height=1.25in,clip,keepaspectratio]{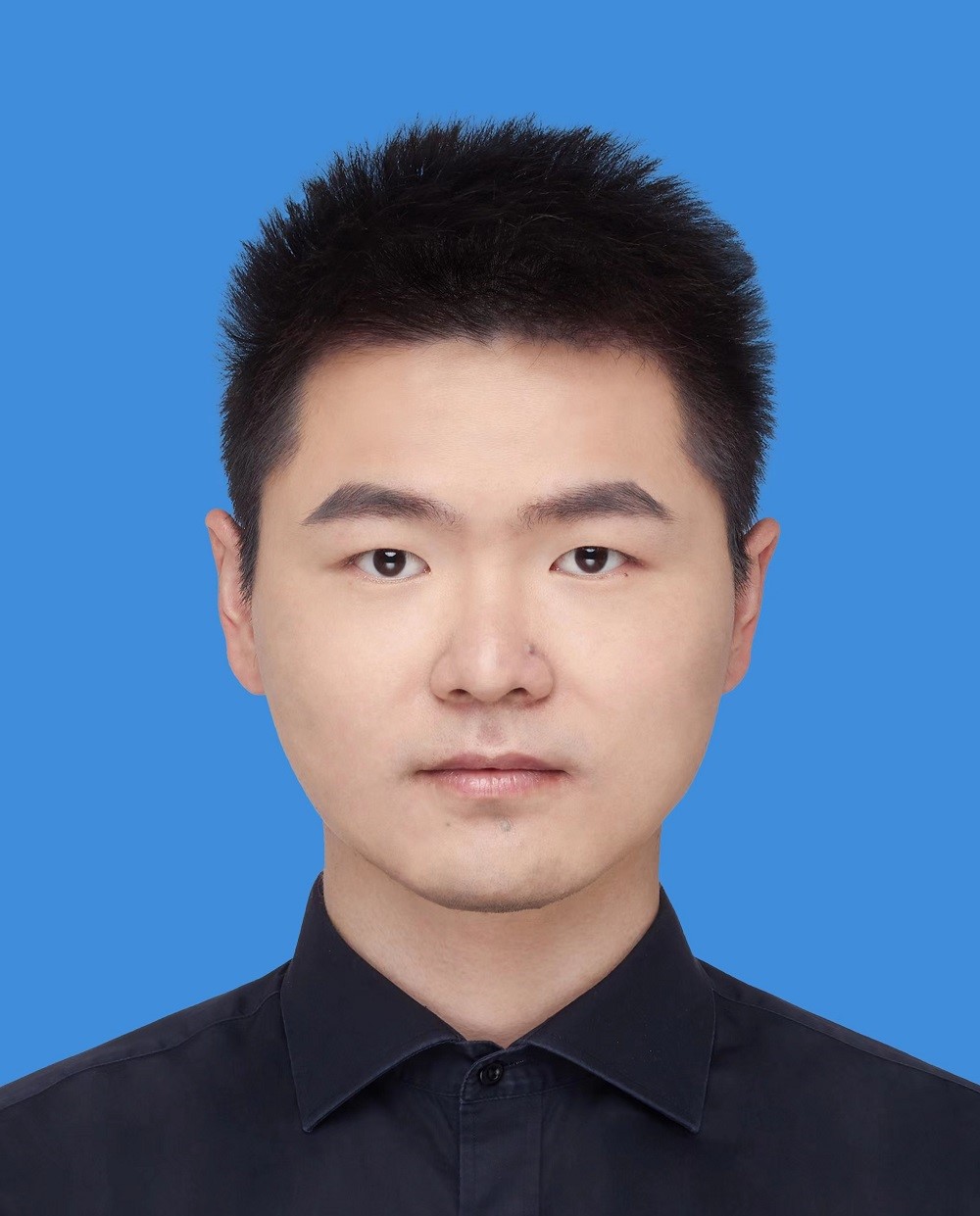}}]
{Wenhan Cao} received his B.E. degree in the School of Electrical Engineering from Beijing Jiaotong University, Beijing, China, in 2019.

He is currently a Ph.D. candidate in the School of Vehicle and Mobility, Tsinghua University, Beijing, China. His research interests include optimal filtering and reinforcement learning. He was a finalist for the Best Student Paper Award at the 2021 IFAC MECC.
\end{IEEEbiography}

\vspace{-0.5cm}

\begin{IEEEbiography}[{\includegraphics[width=1in,height=1.25in,clip,keepaspectratio]{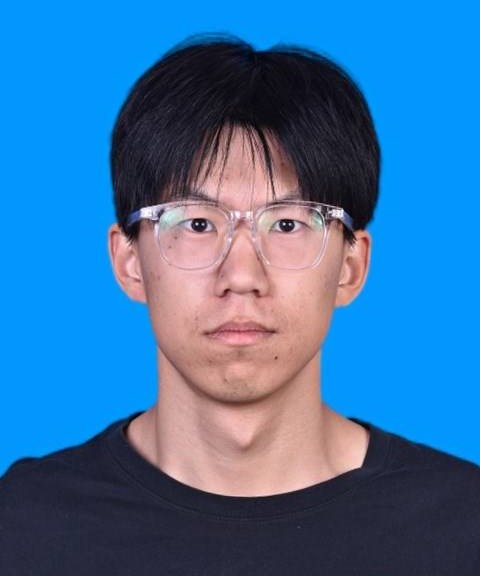}}]
{Tianyi Zhang} received his B.E. degree in the school of Automation Science and Electrical Engineering from Beihang University, Beijing, China, in 2024. He is currently a Ph.D. candidate in the School of Vehicle and Mobility, Tsinghua University, Beijing, China. His research interests include optimal state estimation, Bayesian inference, and reinforcement learning.
\end{IEEEbiography}

\vspace{-0.5cm}

\begin{IEEEbiography}[{\includegraphics[width=1in,height=1.25in,clip,keepaspectratio]{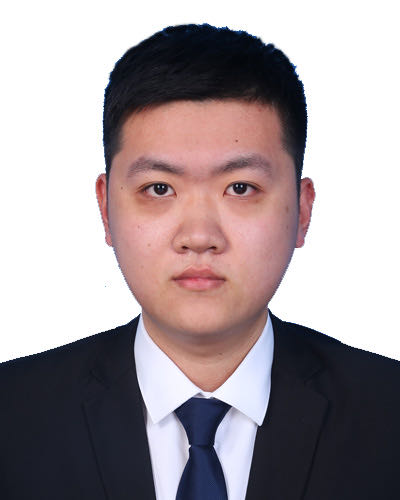}}]
{Zeju Sun} received the B.S. degree from Department of Mathematical Sciences, Tsinghua University, Beijing, China, in 2020. He is currently pursuing the Ph.D. degree in applied mathematics with the department of Mathematical Sciences, Tsinghua University, Beijing, China. His research interests include control theory and nonlinear filtering.
\end{IEEEbiography}

\vspace{-0.5cm}

\begin{IEEEbiography}[{\includegraphics[width=1in,height=1.25in,clip,keepaspectratio]{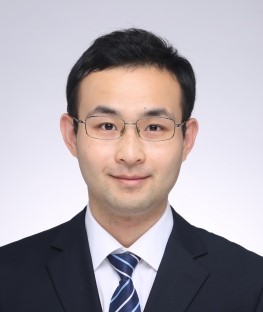}}]
{Chang Liu} (Member, IEEE) received the B.S. degrees in Electronic Information Science and Technology and in Mathematics and Applied Mathematics (double degree) from the Peking University, China, in 2011, and the M.S. degrees in Mechanical Engineering and in Computer Science, and the Ph.D. degree in Mechanical Engineering from the University of California, Berkeley, USA, in 2014, 2015, and 2017, respectively. 
He is currently an Assistant Professor with the Department of Advanced Manufacturing and Robotics, College of Engineering, Peking University. From 2017 to 2020, he was a Postdoctoral Associate with the Cornell University, USA. He has also worked for Ford Motor Company and NVIDIA Corporation on autonomous vehicles. His research interests include robot motion planning, active sensing, and multi-robot collaboration.
\end{IEEEbiography}

\vspace{-0.5cm}

\begin{IEEEbiography}[{\includegraphics[width=1in,height=1.25in,clip,keepaspectratio]{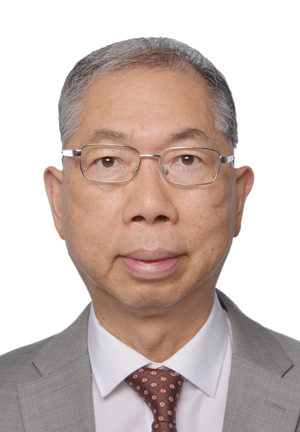}}]{Stephen S.-T. Yau}
 (Life Fellow, IEEE) received the Ph.D. degree in mathematics from the State University of New York at Stony Brook,   NY,   USA in 1976.
 
 He was a Member of the Institute of Advanced Study at Princeton from 1976-1977 and 1981-1982,   and a Benjamin Pierce Assistant Professor at Harvard University during 1977-1980. After that,   he joined the Department of Mathematics,   Statistics and Computer Science (MSCS),   University of Illinois at Chicago (UIC),   and served for over 30 years,   During 2005-2011,   he became a joint Professor with the Department of Electrical and Computer Engineering at the MSCS,   UIC. After retiring in 2011, he joined the Department of Mathematical Sciences at Tsinghua University in Beijing, China, where he served for over 10 years. In 2022, he became a research fellow at the Beijing Institute of Mathematical Sciences and Applications (BIMSA) in Beijing, China, to begin his new research. His research interests include nonlinear filtering,   bioinformatics,   complex algebraic geometry,   CR geometry and singularities theory.
 
 Dr. Yau is the Managing Editor and founder of the \emph{Journal of Algebraic Geometry} since 1991,   and the Editor-in-Chief and founder of \emph{Communications in Information and Systems} from 2000 to the present. He was the General Chairman of the IEEE International Conference on Control and Information,   which was held in the Chinese University of Hong Kong in 1995. He was awarded the Sloan Fellowship in 1980,   the Guggenheim Fellowship in 2000,   and the AMS Fellow Award in 2013. In 2005,   he was entitled the UIC Distinguished Professor.
\end{IEEEbiography}

\vspace{-0.5cm}

\begin{IEEEbiography}[{\includegraphics[width=1in,height=1.25in,clip,keepaspectratio]{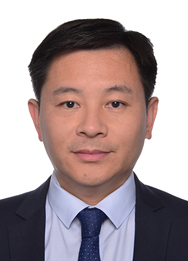}}]
{Shengbo Eben Li} (Senior Member, IEEE) received his M.S. and Ph.D. degrees from Tsinghua University in 2006 and 2009. Before joining Tsinghua University, he has worked at Stanford University, University of Michigan, and UC Berkeley. His active research interests include intelligent vehicles and driver assistance, deep reinforcement learning, optimal control and estimation, etc. He is the author of over 190 peer-reviewed journal/conference papers, and co-inventor of over 40 patents. Dr. Li has received over 20 prestigious awards, including Youth Sci. \& Tech Award of Ministry of Education (annually 10 receivers in China), Natural Science Award of Chinese Association of Automation (First level), National Award for Progress in Sci \& Tech of China, and best (student) paper awards of IET ITS, IEEE ITS, IEEE ICUS, CVCI, etc. He also serves as Board of Governor of IEEE ITS Society, Senior AE of IEEE OJ ITS, and AEs of IEEE ITSM, IEEE TITS, IEEE TIV, IEEE TNNLS, etc. 
\end{IEEEbiography}

\end{document}